\documentclass[journal,comsoc]{IEEEtran}
\usepackage{hyperref}
\usepackage[ruled]{algorithm2e}
\renewcommand{\algorithmcfname}{ALGORITHM}
\usepackage{amsmath,amssymb,amsthm,bbm}
\usepackage{url}
\usepackage{color} 
 \usepackage{booktabs}

\usepackage{graphicx}
\usepackage{caption}
\usepackage{subfigure}

\usepackage{changes}
\usepackage{lipsum}
\definechangesauthor[name={Per cusse}, color=orange]{per}
\setremarkmarkup{(#2)}

\usepackage{color}

\usepackage{epstopdf}
\newtheorem{definition}{Definition}
\newtheorem{lemma}{Lemma}
\newtheorem{theorem}{Theorem}
\newtheorem{proposition}{Proposition}

\makeatletter

\newcommand{\Rmnum}[1]{\expandafter\@slowromancap\romannumeral #1@}
\makeatother

\SetArgSty{textrm}  
\SetAlFnt{\small}
\SetAlCapFnt{\small}
\SetAlCapNameFnt{\small}
\SetAlCapHSkip{0pt}
\IncMargin{-\parindent}

\usepackage[T1]{fontenc}

\usepackage{amsmath}
\usepackage[cmintegrals]{newtxmath}
\begin{document}
%
\title{Tiered cloud storage via two-stage, latency-aware bidding}

\author{Yang Zhang, Arnob Ghosh, Vaneet Aggarwal, and Tian Lan
\thanks{Y. Zhang, A. Ghosh, and V. Aggarwal are with the School of Industrial Engineering, Purdue University, West Lafayette IN 47907, email: \{zhan1925, ghosh39, vaneet\}@purdue.edu. T. Lan is with the  Department of Electrical and Computer Engineering, George Washington University, Washington D.C. 20052, email: tlan@gwu.edu. }}

\maketitle


\begin{abstract}
In cloud storage, the digital data is stored in logical storage pools, backed by heterogeneous physical storage media and computing infrastructure that are managed by a Cloud Service Provider (CSP). One of the key advantages of cloud storage is its elastic pricing mechanism, in which the users need only pay for the resources/services they actually use, e.g., depending on the storage capacity consumed, the number of file accesses per month, and the negotiated Service Level Agreement (SLA). To balance the tradeoff between service performance and cost, CSPs often employ different storage tiers, for instance, cold storage and hot storage. Storing data in hot storage incurs high storage cost yet delivers low access latency, whereas cold storage is able to inexpensively store massive amounts of data and thus provides lower cost with higher latency.

In this paper, we address a major challenge confronting the CSPs utilizing such tiered storage architecture - how to maximize their overall profit over a variety of storage tiers that offer distinct characteristics, as well as file placement and access request scheduling policies. To this end, we propose a scheme where the CSP offers a two-stage auction process for (a) requesting storage capacity, and (b) requesting accesses with latency requirements. Our two-stage bidding scheme provides a hybrid storage and access optimization framework with the objective of maximizing the CSP's total net profit over four dimensions: file acceptance decision, placement of accepted files, file access decision and access request scheduling policy. The proposed optimization is a mixed-integer nonlinear program that is hard to solve. We propose an efficient heuristic to relax the integer optimization and to solve the resulting nonlinear stochastic programs. The algorithm is evaluated under different scenarios and with different storage system parameters, and insightful numerical results are reported by comparing the proposed approach with other profit-maximization models. We see a profit increase of over 60\% of our proposed method compared to other schemes in certain simulation scenarios.

\end{abstract}

\begin{IEEEkeywords}
Cloud Storage, Latency Dependent Pricing, Two-stage bidding, Mixed-Integer nonlinear Programming.
\end{IEEEkeywords}

%
\IEEEpeerreviewmaketitle

\section{Introduction}
\textcolor{black}{
\subsection{ Motivation}}

\IEEEPARstart{T}{he} demand for online data storage is increasing at an unprecedented rate due to growing trends such as cloud computing, big data analytics, and E-commerce activities   \cite{ad09}, and recently by the rise of social networks.
%
%
Cloud storage service is now provided by multiple cloud service providers (CSP) such as Amazon's S3, Amazon's Cloud drive, Dropbox, Google Drive, and Microsoft Azure\cite{94}. \textcolor{black}{Amazon S3 offer 3 major storage classes for different use cases: i) Amazon S3 Standard for general-purpose storage of frequently accessed data; ii) Amazon S3 Standard for Infrequent Access for long-lived, but less frequently accessed data, and iii) Amazon Glacier for long-term archive, while Dropbox has a simple pricing framework, providing two types of storage (Standard and Advanced) for individuals and one type for enterprises.}

Many cloud storage service providers offer throughput or IOPS (Input and Output operations Per Second) guarantees; however, as a large number of files is stored, the latency of accessing stored files becomes an important criterion to evaluate the effectiveness of these storage services. \textcolor{black}{However, there is no consideration of the latency of accessing the stored data in the above pricing schemes in a shorter time scale. Thus, if a user needs to access the file quite often with a lower latency it may not be able to get that service on a given day. }


\textcolor{black}{
\subsection{Market Architecture}}

\textcolor{black}{We consider a two-stage market for providing lower latencies to the users for a certain period in an auction mechanism. All the files are stored in the back-up storage. In stage 1, the users bid in order to store their files in the cold or hot storage which are faster compared to the back-up storage. Additionally, there is a second market which runs more frequently compared to the first stage where the users can update their bids if they require low latency or faster access. Using our two-stage bidding platform, the CSP  can maximize the total profits over file storage and file access while meeting the users' access requirements. On the other hand,  the users with lower latency requirements will be able to get their required quality of service. Thus, the users' utilities will be maximized. }


\textcolor{black}{
\subsection{Challenges}}
One major challenge confronting the service providers these days is: given the price customers are willing to pay, and the expectation of future access rates, how can a service provider maximize its overall profit over a variety of file storage decisions, file access decisions, and access request scheduling policies. Further, they should also ensure that they provide a reliable, efficient storage that meets customer's latency requirements. This challenge necessities novel pricing mechanisms that go beyond existing approaches such as resource-based pricing, usage-based pricing, time-dependent pricing in cloud computing and online storage.

\textcolor{black}{
\subsection{Contribution}}
In order to store the files in the {\em cold storage} or {\em hot storage}, we propose a systematic framework for two-stage, latency-dependent bidding, which aims to maximize the cloud storage provider's net profit in tiered cloud storage systems where tenants may have different budgets, access patterns and performance requirements as described in Section~\ref{sec:system}. The proposed two-stage, latency-aware bidding mechanism works as follows.  The cloud service provider (CSP) has two tiers of storage: hot storage and cold storage with different service rates.   Users can bid for storage and access, in two separate stages, without knowing how the CSP stores the contents.   In the first stage (\textit{request for storage}), the user specifies storage size, expected access rates, and latency requirements.  If the CSP decides to accept the bid, it will place two copies of data:  one in the cold storage and another one in either the hot storage or cold storage.  In the second stage (\textit{request for access}), the CSP can decide whether to accept the access requests based on the bid and where to retrieve the files from to meet the access latency requirements. The second-stage auction runs on a shorter time scale (every hour) and the first-stage auction runs on a longer time scale (every day) since the access pattern of files changes faster.

%

The second-stage decision inherently depends on the first-stage decision. For example, if the CSP decides to store both the original file and its copy in cold storage, the file can be accessed from the cold storage only. However, if accessing from cold storage does not meet the access latency requirement (storage servers might get congested due to high request arrival rates and low service rates),  the CSP may not be able to serve the request at once. In this case, the CSP will lose profit due to the loss of the access bids from the users.
The optimal first-stage decision decision inherently depends on the second-stage decision. For example, if a user bids at a low price for storage, the file may be stored in the cold storage; however, the user may then bid at a higher price with lower latency requirement in the second stage.  In that case, its bid may not be accepted as the latency requirement may not be matched because the file was stored in the cold storage in the first place.  Unfortunately, the access bids, latency requirements, and the access arrival rates all are random variables, and the realization of these random variables are not known beforehand.

We first formulate the second-stage decision problem whether to accept the bids and scheduling decision (whether to access the file from the cold or hot storage) given a first-stage decision as an integer programming problem with non-convex constraints. Since the second stage parameters are random, we consider multiple random realizations of these variables and average the objective function over these realizations (or scenarios). We then formulate the first-stage decision problem as a deterministic equivalent program where we maximize the profit from the storage and the expected second stage profit while satisfying the latency requirements for each scenario (Section~\ref{sec:formulation}). However, the problem again turns out to be an integer programming with non-convex constraints. We first relax the integer constraints by using sigmoid function as the penalty, which closely matches the required penalty function. The relaxed problem is smooth and we can obtain a local solution using the KKT conditions. The solution of the relaxed problem is then converted to the nearest integers. Because of the sigmoid function, the solution attained by the relaxed problem and the feasible one is quite close. In Section~\ref{sec:simulation}, we show the strength of our proposed method in achieving significantly higher profit as compared to the other algorithms which do not consider the second stage recourse decision while taking the first-stage decision. 

Our solution exploits a number of key design tradeoffs. First, any efficient cloud storage and access strategies must meet both the service provider's constraints and customers' requirements. The constraints from the service provider might come from tiered cloud storage architecture, storage-related costs, reliability level and capacities of each tier of storage. The requirements from customers include bidding prices of storage and access, latency requirements and expected access request arrival rates. Second, while placing as much content as possible in cold storage could potentially reduce storage cost, it may be insufficient to meet clients' latency requirements. On the other hand, although storing more content in hot storage improves service latency, it results in higher storage price, which might cause customer churn. A solution exploiting this tradeoff is thus necessary to determine the optimal placement (and duplication strategy) of files in tiered storage. As a result, jointly scheduling all the file access requests to avoid congestion in each storage tier becomes challenging and must take into account the impact of request patterns and access decisions of all clients.

The main contribution in this paper can be summarized as follows:

   \noindent 1. \emph{Comprehensive future consideration}: This paper aims to propose a systematic framework that integrates both file storage and file access, which optimizes the system over four dimensions: file acceptance decision, placement of accepted files, file access decision and access request scheduling. The proposed framework encompasses future access information such as bidding price for access, latency requirements and expected access request arrival rates.

  \noindent 2. \emph{Two-Stage, Latency-Aware Bidding}:  Most storage pricing schemes consider both storage and access at the same time; our scheme is novel as it allows users to bid for storage and access (with latency requirements) separately and gives the CSP more flexibility in optimizing the tiered-storage to maximize the profits.

   \noindent 3. \emph{Computational Efficiency}:  We quantify the service latency with respect to both hot and cold storage. The proposed optimization is modeled as a mixed-integer nonlinear program (MINLP), which is hard to solve. We propose an efficient heuristic to relax the integer optimization and solve the non-convex problem.

 \noindent 4. \emph{Insightful Numerical Results}: The performance of the proposed approach is evaluated in various cases. It is observed that the profits obtained from the proposed method are higher than those of other methods, and the access request acceptance rate (ARAR) also dominates that of other methods as the capacity of the cold storage or the service rate of hot storage increases.  For example, we see a profit increase of over 60\% of our proposed method compared to other schemes  as the capacity of cold storage increases beyond 500TB with our simulation scenario.

 The rest of the paper is organized as follows. Section \ref{sec:related} describes the related work. The system model for the tiered architecture and the two-stage auction framework is described in Section \ref{sec:system}, and the two-stage optimization problem is formally defined in Section \ref{sec:formulation}. Section \ref{sec:algo} gives the proposed solution for the mixed integer non-linear program and Section \ref{sec:simulation}  validates our proposed policy and evaluates its performance using numerical studies. Finally, Section \ref{sec:concl} presents our conclusions.  

\section{ Related Literature}\label{sec:related}
 Tiered storage has been used in many contexts so as to achieve better cost-performance tradeoffs by placing the workload on a hybrid storage that includes multiple hot and cold storage tiers   \cite{gu11}  \cite{kim14}  \cite{li14}  \cite{wang14}  \cite{oa11}  \cite{i12}  \cite{tier_store}. However, the pricing solution for multi-tier cloud storage is quite limited to resource/usage-based pricing, as shown in   \cite{Naldi13}. Some of the recent pricing schemes for online storage providers include those  AWS S3, Dropbox, Google Drive, etc. and their  current pricing plans can be found at   \cite{amazonaws},   \cite{drop},   \cite{google}, respectively. Typically, they often offer a flat price for the storage service with a limited storage capacity or access rates. For example, Amazon provides three types of storage facilities depending on the access rates. However, our model is different from the existing practices. First, we consider a two-stage auction model where in the first-stage, the users can move its file to (tiered) cold/hot storage by adjusting their bids. In the second-stage, the users bid to access the files. Note that the first-stage auction is run once in a day (or week), while the second-stage once an hour (or day). Thus, it provides a greater flexibility to the users to adjust their bids according to their daily requirements. In contrast, the user has to pay a flat rate price for a month if one wants to achieve a faster access rate in the Amazon. Second, in contrast to the pricing mechanisms of Amazon and Dropbox, we consider the latency requirements of the users while accepting the bids even at the first-stage. 

Pricing for cloud computing has been widely studied  \cite{xu,ebay,double_auction,random_auction,online_auction,lin}.
\textcolor{black}{
Game Theory and Auctions are broadly adopted as mechanisms for cloud service. For example, in \cite{chris16}, a game theoretical model is used to induce a truthful  cloud storage selection mechanism where the service providers bid the quality of service; in \cite{zhou15} , an online procurement auction mechanism is proposed to maximize the long-term social welfare; A Vickery Clarke Grove (VCG) auction-based dynamic pricing scheme is proposed for cloud services in \cite{wu17}. Recently, a stackelberg game model is proposed in \cite{ma16} to derive the pricing scheme. The stackelberg game consists of two stages-- i) in the first stage, the service provider determines a price which is both time and location dependent, ii) in the second stage, the users decide the schedule of the mobile traffic depending on the prices. However, compared to the above papers, we consider a scenario where the users bid in a two-stage-- in the first stage, the users bid in order to store their files in the hot or cold storage; in the second stage, the users again bid for the latency requirements and the access arrival requests. The cloud service provider in the first stage is unaware of the bids of the users in the second stage. However, the optimal decision is inherently depends on the second stage decisions. Thus, the problem is inherently challenging , and turns out to be a non-convex mixed integer problem.}

{\em To the best of our knowledge, such kind of auction mechanisms have not been considered in the literature yet.} Additionally, the above papers mainly considered Vickrey-Clarke-Groves (VCG) type auctions \cite{varian2014vcg} or their variants. However, our problem turns out to be a complex non-convex optimization problem. A VCG-type auction will have high complexity and the optimality cannot be guaranteed because of the non-convexity of the problem.

\section{System Model}\label{sec:system}
\subsection{Tiered Architecture}
We consider a cloud storage provider (CSP) which has a tiered storage architecture. Each file is stored in an inexpensive {\em back-up} storage facility. For example, Amazon Web service (AWS) charges $0.023$ per GB per month for standard storage. The back-up  storage can be considered to consist of hard disk drive (HDD) which is inexpensive, but, the service rate is slow and unreliable. Since it is inexpensive, the latency cannot be guaranteed as a lot of files can be stored. In order to provide a faster service the CSP can offer two types of storage -- i){\em cold storage} and ii) {\em hot storage}. Cold Storage is made of SSHD (combination of solid state drive (SSD) and HDD) which is expensive compared to the HDD, however, the service rate is faster and there is more reliability against disk failure. The hot storage is the most expensive one as it is made of SSD, however, the service rate is also the fastest. Thus, if files are stored in the hot storage, they will have faster access. 

\subsection{Two-Stage Auction Framework}
\begin{figure}
\centering
\includegraphics[trim=3.8in 2.8in 2.2in .7in, clip, width=.45\textwidth]{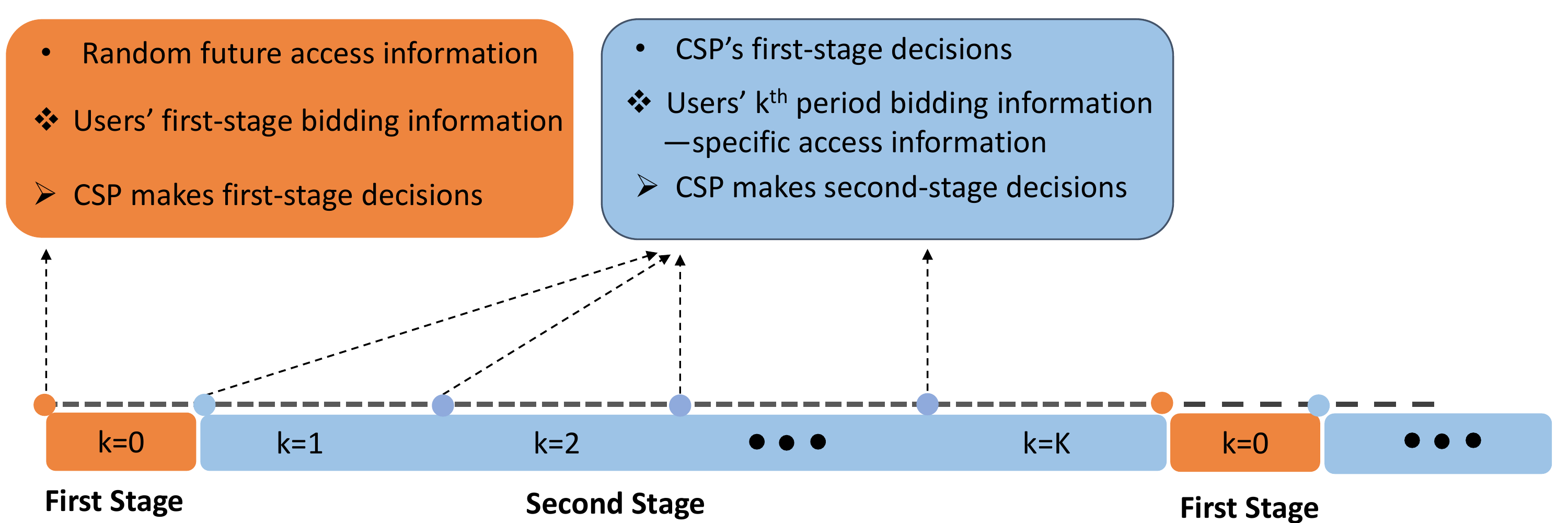}
\caption{Two-Stage Auction Framework}
\label{fig0}
\end{figure}

In order to store the files in the {\em cold storage} or {\em hot storage}, the CSP will operate a market. In the {\em first stage}, the users \footnote{We denote all the clients of cloud service providers as users. Thus, users may be the individuals, enterprises, or organizations}bid to store their files in the upgraded storage facilities. The CSP decides whether to accept the file and where to store the original file and its copies. We  model the storage platform as providing dual replication of files, so each file has a duplicated copy. \footnote{Multiple copies of the file can be created in practice. However, it will increase the storage requirement and the computational complexity of computing the acceptance/rejection of bid, and the access probabilities. The consideration of the scenario where any specific number of copies can be stored is left for future work.}\textcolor{black}{To ensure data durability and availability, data replication is broadly adopted by data center storage systems, such as Hadoop Distributed File System \cite{shva10}, RAMCloud \cite{diego11}, and Google File System \cite{sanjay03}.} If a file is accepted for storage, the user pays the bidding price, otherwise, it pays nothing. The CSP stores one copy in the cold storage. The CSP also decides whether to store the other copy either in the hot or the cold storage.

In the {\em second stage}, the users whose files get accepted for storage, bid again for accessing the files. The CSP needs to decide whether to accept the access requests and if accepted, from where the files should be accessed (either cold or hot storage) in order to meet the access latency requirements.  If the user's request is accepted, it pays the bidding price, otherwise, it pays nothing. The second stage decision inherently depends on the first stage decision. For example, if the CSP decides to store the both the original file and its copy in cold storage, the file can be accessed from the cold storage only. However, if accessing from cold storage does not meet the access latency requirement (storage servers might get congested due to high request arrival rates and low service rates),  the CSP may not be able to serve the request at once. In this case, the CSP will lose profit due to the loss of the access bids from the users. Fig.~\ref{fig0} depicts graphically the major considerations in the two-stage problem.

The {\em optimal} first stage decision of the CSP  inherently depends on the second stage decisions. For example, if the CSP decides to store a file in the cold storage because of its low storage bid, it can bid a high value for the access in the second stage. However, the CSP may not accept the bid because of the lower service rate of the cold storage. Hence, the CSP's profit will be reduced. These access bids, the access arrival requests, and the latency requirements are random variables which cannot be known during the first stage decision process which makes finding an optimal first stage decision is inherently difficult. We assume that the two bidding stages take place at different time-scales. In particular, while users' files typically remain in the storage system for a long time period (e.g., a day, or several days in Stage 1), the latency-dependent file access decisions (in Stage 2) can be adjusted more frequently on a much smaller time-scale (e.g., every hour), e.g., during busy and off-peak hours. Intuitively, the user's need to access  a file changes on a shorter time scale compared to its storage decision. Hence, the second stage auction must be  run more frequently. Note that the frequency of the second stage auction can be changed depending on the change of the access request rates of some of the files. 

 Note that not all the access bids of the files stored in the cold or hot storage will be accepted. The acceptance depends on the access bids and the latency requirements. However, it is still useful for the user to participate in the first stage {\it i.e.}, paying a higher price to store its file in the cold or hot storage. This is because the user may have to access the file only for a certain number of hours in a day, the user can participate in the first stage auction where its files will be stored either in the cold storage or the hot storage at the start of the day. When the user needs to access the file, it bids in the second stage auction. Note that since both the cold storage and the hot storage have higher service rates as compared to the back-up storage, the users can access files at a much faster rates compared to the traditional back-up storage even if their access bids are not accepted at all.
 
 

\textcolor{black}{Also note that we have a back-up storage for all the files. Initially, all the files are stored in back-up storage. The users then bid in order to store their files slightly faster cold storage or the fastest hot storage.  After this fist-stage auction, the files that are accepted will be copied and moved to store in cold or hot storage. However, the rest of files will be stored in the back-up storage. If the user's storage bid is rejected, she will still be able to access those files from the back-up storage.  Our second stage bidding is only designed for premium data access, while a standard, basic service to access data is provided to all files stored in the system. If a user's access bid is not accepted by CSP in the second stage, she still be able to access the file from hot or cold storage. \emph{Thus, service availability is indeed guaranteed.} However, there will be no guarantee on the latency or the speed of accessing the files in the above two cases.}

\section{Problem Formulation}\label{sec:formulation}
In this section, we formally define the two-stage optimization problem. 

\subsection{First-Stage Decision}

In the first stage, the CSP decides -- i) whether to store a file or not, ii) if it decides to store the file whether to keep the duplicated copy of the file in the hot storage or cold storage (Original copy of an accepted file is always stored in cold storage).  {\footnote{\begin{color}{black} Note that we consider storing the first copy in cold storage due to its relatively low cost and large capacity, while the analysis and optimization remain the same if it is replaced by any other type of storage tier. \end{color}}} Let $I$ be the total number of files participate in the first-stage auction.

We consider a first price auction where the user pays the price it bids. This auction is run once in a day or once in a week. 
Let $A_{i}=1$ denote that the file $i=1,\ldots, I$ is accepted for storage; $A_i=0$ if it is not accepted. Let $R_{i}=1$  denote that the copy of the  file $i=1,\ldots, I$ is stored in the hot storage; otherwise, $R_i=0$. Note that if $A_{i}=0$, then the file is not stored anywhere, thus, $R_{i}=0$. However, if $A_{i}=1$, $R_{i}$ can be either $1$ or $0$. Nonetheless, if $R_{i}=1$,  $A_{i}$ must be $1$.

Also note that if $R_{i}=0$, there are two possibilities: (i) $A_i=1$, thus, both of the original and duplicated copies will be stored in cold storage (hence the number copies of file $i$ stored in the cold storage is $2$); or, (ii) the file storage bid is rejected ($A_i=0$). {\em Therefore, the number of copies of file $i$ stored in the hot storage and cold storage is $R_{i}$ and $2A_{i}-R_{i}$ respectively.}

Let $S_i$ be the size of the file $i$ and $C_j$ be the capacity of storage $j$, where $j=1$ denotes the cold storage and $j=2$ denotes the hot storage. Since the total stored files cannot exceed the capacity,
\begin{eqnarray}
& \sum\limits_{i=1}^{I}S_{i}(2A_{i}-R_{i})\leq C_{1}
\label{green-constraint-1}
\\
& \sum\limits_{i}S_{i}R_{i}\leq C_{2}
\label{green-constraint-2}
\end{eqnarray}
Since $A_i$ must be $1$ if $R_i$ is 1, 
\begin{eqnarray}
& A_{i}-R_{i}\geq 0,  \quad \forall i
\label{green-constraint-3}
\end{eqnarray}

\subsection{Second Stage Decision Problem}
After storing the files, the users bid for accessing the files in $T$ different time slots. While bidding, the user also gives the access request arrival rates and the latency requirements in each slot. This market is run on a shorter time scale ({\it e.g.}, the duration can be an hour or half an hour). The user can update its bid at different time slots depending on its requirements.

The access request arrival rates, the access bid prices, and the latency requirements are {\em random variables}, which are governed by the user's requirements. We assume that the random variables can be modeled by $K$ realizations of the random variables ( or $K$ scenarios).  \textcolor{black}{The decision is time-based. The second stage runs at epochs $t = 1, 2, ... (T-1), T,$ (e.g., every hour), and different scenarios, i.e., the user's bid, latency and arrival rate, can vary over these epochs. We take access decisions at each epoch based on the bids. The first stage runs at every $T$ periods, (e.g., every day), and we take storage decision by considering the possible scenarios and the associated probabilities across all the $T$ periods. }

\textcolor{black}
{ Workloads for accessing data follow some pattern \cite{di11} \cite{gon11} \cite{dan14}. However, the CSP is unaware of the exact joint distribution function of the bidding prices, access arrival request rates, and the latency requirements. However, in the scenario-based approach, we do not need to know ay specific distribution function. Specifically, we can generate the empirical distribution from the bidding history. For example, from Fig. \ref{fig0} we know that the first stage auction runs in a longer time scale (e.g., every day) and second stage auction runs in a shorter time scale (e.g., every hour). Then in the following day, the CSP can learn the (joint) empirical distribution of bid price, latency and arrival rate based on the access information from the last (few) day(s). Thus, our approach can be applied to any scenario where a workload pattern does not need to be learnt.} 

For each scenario $k=1,\ldots, K$, we denote the latency requirement of file $i$ as $l_i^{k}$, the access bid price as $q_i^{k}$, and the access request arrival rate as $\lambda_i^{k}$. The scenarios can be generated from the past history of the user's data. We assume that scenario $k$ occurs with probability $p^{k}$. {\em We do not put any restriction on the dependence of the bids, the latency bids, and the bids. Specifically, they can be obtained from a joint distribution.} However, the CSP is unaware of the distribution. It learns from the bidding history and updates the set of scenarios. 

\color{black}
\subsubsection{Access arrival rates}

\textcolor{black}{The access requests are independent and in a certain time slot, the number of these requests are integer and can be considered independent of the past requests. It is often assumed that the inter arrival time follows exponential distribution \cite{94}\cite{kristen15}\cite{gauri15}. Thus, we consider a Poisson arrival process. We use M/G/1 queuing model.{\footnote{\begin{color}{black} In this paper, we consider one disk for each hierarchy. Because the multiple servers will reduce the bandwidth of each server.  Thus, the capacity of serving requests from each server will be reduced.  As a result, the latency of each request will be increased.\end{color}}} M/M/1 or M/G/1 queuing model is also used in \cite{94}\cite{kristen15}\cite{gauri15}. }

\color{black}
\subsubsection{Access request acceptance}
The CSP decides whether to accept the bid of the access request of each and every file. Let $H_i^{k}$ denote the decision that whether the file $i$ is accepted in scenario $k\in \{1,\ldots, K\}$. $H_i^{k}=1$ indicates that the access bid is accepted; $H_i^{k}=0$ indicates that the access bid is rejected.

Note that when the second stage decision is taken, the first-stage decision variables $A_i$ and $R_i$ are known. If $A_i=0$, then $H_i^{k}=0$ for all $k\in \{1,\ldots, K\}$  since file $i$ is not stored in the cold or hot storage, then its access bid cannot be accepted. On the other hand if $A_i=1$, $H_i^{k}$ can be either $0$ or $1$. This is because even if $A_i=1$, it cannot be guaranteed that the access bid will be accepted in scenario $k$. The access bid will be accepted based on how much profit will be  made and whether the latency requirement can be satisfied by accepting the bid. Hence,
\begin{align}\label{green-constraint-8}
H_i^{k}\leq A_i \quad\forall i.
\end{align}


\subsubsection{Probabilistic Scheduling}
\textcolor{black}{Probabilistic Scheduling has been successfully applied in display ad allocation problem on the Internet \cite{hae11} and  high--aggregate bandwidth switches \cite{gia03}. Such a strategy has been also shown to be nearly optimal in cloud storage \cite{94}. }

Recall that if file $i$  is accepted for storage, the original copy would be stored in cold storage and the duplicated copy will be either stored in cold storage ($R_i=0)$ or in hot storage ($R_i=1$). As we have copies of a file in both hot and cold storage in the latter case ($R_i=1$), the CSP needs to decide where the file should be accessed according to its bidding price and latency requirement. In probabilistic scheduling, each request for file $i$ has a certain probability to be scheduled to each storage $j$. For the $k$-th scenario, we have to decide $0\leq \pi_{i,j}^{k}\leq 1$ which denotes the probability that the file $i$ will be fetched from storage $j$, $j\in \{1,2\}$ for the $k$-th scenario. Intuitively, $\pi_{i,j}^{k}$ denotes how often the file $i$ should be fetched from storage $j$ for scenario $k$. Needless to say, if $R_i=0$, then $\pi_{i,2}^{k}=0$ for all $k\in \{1,\ldots, K\}$. Hence,
\begin{align}\label{green-constraint-11}
0\leq\pi_{i,2}^{k}\leq R_{i}, \quad \forall i.
\end{align}

And $\pi_{i,j}^{k}=0$ for files which have not been accepted for access requests. Thus,
\begin{eqnarray}
&   \sum\limits_{j=1}^{2} \pi_{i,j}^{k}=H_{i}^{k}, \quad \forall i.
\label{green-constraint-9}
\end{eqnarray}

Recall that $\lambda_i^{k}$ denotes the access request arrival rate of file $i$ in scenario $k$ within a slot. Thus, the total expected file access request rates for file $i$  to storage $j$ in the $k$-th scenario within the slot is given by $\lambda_i^{k}\pi_{i,j}^{k}$. The total expected  file access request rate to storage $j$ must be less that the file service rate (in Mb/s) of storage $j$; otherwise, the queue length will be $\infty$ and the storage $j$ cannot handle requests. Hence,
\begin{eqnarray}
& \sum\limits_{i} \lambda_{i}^{k}\pi_{i,j}^{k}S_{i}<\mu_{j}, \quad \forall j.
\label{green-constraint-7}
\end{eqnarray}

\textcolor{black}{Our scheduling approach will be optimal in an expected sense. However, the scheduling approach may be sub-optimal for a given scenario. Obtaining an optimal deterministic schedule is a NP-hard problem in general for a given scenario.}


\subsubsection{Latency Analysis}
\begin{definition}\label{defn:latency}
Latency is the sum of the time a file access request spends in the queue for service (waiting time) and the service time.
\end{definition}
The users strictly prefer a low latency. Studies show that in internet application even $0.1$s increase in the latency can significantly reduce the profit \cite{latency}.
The latency for file $i$ will inherently depend on the probabilistic scheduling decision $\pi_{i,j}^{k}$, arrival rate $\lambda_{i}^{k}$, and the service rate of the storage $\mu_j$. \textcolor{black}{Given the same number of files with same sizes are being served, a file will spend less time for service because of the higher service rate of the hot storage compared to the cold storage. Thus the latency will be shorter in the hot storage.} However, a user may have to pay more for accessing.
In the following, we provide the expression for the expected latency of a file.  Before that, we introduce a notation which we use throughout.
\begin{definition}
Let $\bar{T}_i^{k}, k=1,\ldots, K$ denote the expected latency for file $i$ request at scenario $k\in \{1,\ldots,K\}$.
\end{definition}
Let $Q_{j}^{k}$ denote the waiting time at storage $j, j=1,2$ for scenario $k$. Recall that $\pi_{i,j}^{k}$ denotes the probability with which  file $i$ will be fetched from storage $j$ in scenario $k$. Hence, the expected waiting time for file $i$ at scenario $k$ is $\sum\limits_{j}\pi_{i,j}^{k}\mathrm{E}[Q_{j}^{k}]$.
 Recall that $\mu_j$ is the service rate in Mb/s for storage $j$. Since the size of the file $i$ is $S_i$ and the probability that the request for file $i$ will be sent to storage $j$ at scenario $k$ is $\pi_{i,j}^{k}$, the expected service time for file $i$ in scenario $k$ is
\begin{align}
\sum\limits_{j}\frac{\pi_{i,j}^{k}S_i}{\mu_j}
\end{align}
From Definition~\ref{defn:latency} we have

\begin{equation} \label{eq3}
\bar{T}^{k}_{i}= \sum\limits_{j}\pi_{i,j}^{k}\mathrm{E}[Q_{j}^{k}]+\sum\limits_{j}\frac{\pi_{i,j}^{k}S_i}{\mu_j}
\end{equation}

The next result characterizes $\mathrm{E}[Q_{j}^{k}]$.

\begin{theorem}\label{lm:waitingtime}
The mean waiting time at storage $j$ for scenario $k$, $\mathrm{E}[Q_{j}^{k}]$ is given as follows. 
\begin{equation} \label{eq4}
 \mathrm{E}[Q_{j}^{k}] = \frac{ \sum_{i}\lambda_{i}^{k}\pi_{i,j}^{k}S_{i}^{2}   }{ \mu_{j}(\mu_{j}-\sum_{i}\lambda_{i}^{k}\pi^{k}_{i,j}S_{i}   )  }
\end{equation}
\end{theorem}
\begin{proof}
In order to simplify notations, we introduce three auxiliary functions: $ f= \sum_{i}\lambda_{i}^{k}\pi_{i,j}^{k}S_{i}  $, $ g= \sum_{i}\lambda_{i}^{k}\pi_{i,j}^{k}   $, and $ h= \sum_{i}\lambda_{i}^{k}\pi_{i,j}^{k}S_{i}^{2}  $.

Using the moments of the service time given in Appendix \ref{apdx:moments}, we have  $\Lambda_{j}^{k}=g$, $ \mathrm{E}[X_{j}^{k}]=\frac{f}{\mu_{j}g}$, and $\mathrm{E}[(X_{j}^{k})^{2}]= \frac{2h}{ \mu^{2}g }$.  Using Pollaczek-Khinchin formula for M/G/1 queues \cite{chan1997pollaczek},  we have $\mathrm{E}[Q_{j}^{k}]=\frac{ h  }{ \mu_{j}(\mu_{j}-f   )  }$. Expanding the terms, we get the result as in the statement of the Theorem. 
\end{proof}
%
Using Lemma~\ref{lm:statistics} in (\ref{eq3}) (which can be found in Appendix), we have
\begin{equation} \label{eq5}
\bar{T}_{i}^{k}= \sum\limits_{j}\frac{\pi_{i,j}^{k}S_{i}}{\mu_{j} }+ \sum\limits_{j} \pi_{i,j}^{k}  \left( \frac{ \sum\limits_{i}\lambda_{i}^{k}\pi^{k}_{i,j}S_{i}^{2}   }{ \mu_{j}(\mu_{j}-\sum\limits_{i}\lambda_{i}^{k}\pi^{k}_{i,j}S_{i}  )   }\right)
\end{equation}
Note that by differentiating twice one can easily discern that $\bar{T}_i^{k}$ is convex in each $\pi_{i,j}^{k}$. However,   $\bar{T}_{i}^{k}$ is jointly {\em non-convex} in $\pi_{i,j}^{k}$. This is because of the terms $\pi_{i,1}^{k}\pi_{i,2}^{k}$, which is not jointly convex in $\pi_{i,1}^{k}$ and $\pi_{i,2}^{k}$.

Note that the latency depends on the file size $S_i$: if the file size is large, the latency will be large.  Thus, it shows that for the same access bid the files of smaller sizes will be preferred (given that its latency requirement is satisfied) as it will allow the CSP to accept more access requests. Also note that if $\lambda_i^{k}\pi_{i,j}^{k}$ is large for some $j$, then the latency again increases, hence, the latency of storage facility $j$ increases if too many requests are directed towards $j$. Thus, the CSP has to judiciously select $\pi_{i,j}^{k}$. If a large number of requests are directed towards the hot storage, the latency requirement may not be satisfied which may decrease the CSP's profit. Also note that $\bar{T}_i^{k}=0$ if the file is not accepted for accessing. Recall that the latency requirement for file $i$ in scenario $k$ is $l_i^{k}$. Hence, we must have
\begin{align}\label{green-constraint-6}
\bar{T}_i^{k}\leq l_i^{k}
\end{align}

\subsubsection{Second Stage Optimization Problem}
The second stage profit of the CSP if the scenario $k\in \{1,\ldots, K\}$ is realized is given by
\begin{align}
\sum\limits_{i}q_i^{k}H_i^{k}
\end{align}
Recall that $q_i^{k}$ is the access bid for file $i$ in scenario $k$. Hence, the second stage optimization problem if scenario $k$ is realized is given by
\begin{align}
& {\text{(P2) }} \nonumber\\
& {\text{ max}}
& & \sum\limits_{i}q_i^{k}H_i^{k}\nonumber\\
& \text{s.t. }
& &  (\ref{green-constraint-8}), (\ref{green-constraint-11}), (\ref{green-constraint-9}),(\ref{green-constraint-7}),(\ref{eq5}), (\ref{green-constraint-6})\nonumber\\
& & &H_i^{k}\in \{0,1\} \quad \forall i \label{eq:h}\\
& & &\text{var }: H_i^{k},\pi_{i,j}^{k}
\end{align}

Note that if a user bids high for access, but its size is large or the arrival rate is high, then the latency (\ref{eq5}) may increase and the CSP will lose the profit as the CSP may satisfy only few requirements of latencies. \textcolor{black}{Problem (P2) is a integer nonlinear program, which is not trivial to get solved. }Hence, it is not apriori clear that how the CPS should select the access bids.
\subsection{Deterministic Equivalent Program}
Now, we formally formulate the first-stage stochastic program. Let $P_i$ be the bid price of file $i$ for storage. \textcolor{black}{Let $c_1$ and $c_2$ denote the total cost incurred by the service provider for storing a file $i$ the hot and cold storage respectively.} Recall that $j=1$ ($j=2$, resp.) denotes that the storage is cold (hot, resp.). Hence, the profit obtained by the CSP for {\em storage} is
\begin{align}
\sum\limits_{i} P_{i}A_{i}-\sum\limits_{i}S_{i}(2A_{i}-R_{i})c_{1}-\sum\limits_{i}S_{i}R_{i}c_{2}
\end{align}
Since the second stage decision variables inherently depend on the first stage and the CSP wants to maximize the total profit, thus, the CSP needs to consider the second stage decision while taking the first stage decision. Hence, the first stage decision problem is different from the standard knapsack problem.

Note that the CSP knows that the access bid price for file $i$ in scenario $k$ is $q_i^{k}$. Recall that the probability with which scenario $k$ is generated is $p^{k}$. Hence, the expected profit from the second stage decision is
\begin{align}
T\sum\limits_{k=1}^{K}\sum\limits_{i=1}^{I}p^{k}q_i^{k}H_i^{k}
\end{align}
$T$ is the total number of slots where the access auctions are run. In the first stage, the CSP wants to maximize the total expected profit. However, the expected profit also depends on the second stage decision variables. Therefore, we should find $H_i^{k}$ and $\pi_{i,j}^{k}$ for each possible scenario. We formulate the first-stage decision problem as the so-called {\em deterministic equivalent program}  \cite{equivalent_deterministic} in the following:
\begin{align}
& {\text{(P1) }} \nonumber\\
 & \textrm{min}& &\sum\limits_{i} P_{i}A_{i}-\sum\limits_{i}S_{i}(2A_{i}-R_{i})c_{1}-\sum\limits_{i}S_{i}R_{i}c_{2} \\
& && +  T\sum\limits_{k}\sum\limits_{i}p^{k}q_i^{k}H_i^{k}\\
&\textrm{s.t.} & & (\ref{green-constraint-1})-(\ref{green-constraint-3}), A_i\in \{0,1\}, R_i\in \{0,1\}\nonumber\\
&  && \bar{ T_{i}}^{k} \leq l_{i}^{k}, \quad \forall i, \quad \forall k \label{P2-2}\\
& & & \sum\limits_{i} \lambda_{i}^{k}\pi_{i,j}^{k}S_{i}<\mu_{j}, \quad \forall j, \quad \forall k \label{P2-3}\\
& & &  H_{i}^{k}\leq A_{i}, \quad \forall i, \quad \forall k \label{P2-4}\\
& & & \sum\limits_{j} \pi_{i,j}^{k}=H_{i}^{k}, \quad\forall i,  \quad \forall k \label{P2-5} \\
& & &   0\leq\pi_{i,2}^{k}\leq R_{i}, \quad\forall i, \quad \forall k \label{P2-7}\\
& & &  H_{i}^{k} \in \{0,1\}, \quad\forall i, \quad \forall k \label{P2-8}\\
& & &  \text{var: }  A_i, R_i, \pi_{i,j}^{k}, H_i^{k}
\end{align}

Note that the constraints in (\ref{P2-2})-(\ref{P2-8}) are for the second stage decisions. Also note that though we solve for $\pi_{i,j}^{k}$ and $H_i^{k}$, the decision variables are of interest in the first stage, which are $A_i$ and $R_i$. After $A_i$ and $R_i$ are decided, the optimization problem (P2) is solved if scenario $k$ is realized.  In the deterministic equivalent program, the number of scenarios $K$ may be very large which increases the number of constraints and the decision space. One remedy is to discard those scenarios which occur with very low probability.

\textcolor{black}{The CSP is unaware of a specific scenario in the first stage. Thus, we consider that the CSP will decide whether to accept a bid, and storing the file in the cold or hot while maximizing the expected revenue over all the scenarios that can generate in the second stage. Note that in the second stage, the CSP is aware of the bids of the users. Thus, the CSP is aware of the specific scenario while taking the second stage decision. }

\color{black}
\begin{theorem}\label{nonconvex}
Problem (P1) and (P2) are non-convex. 
\end{theorem}
\begin{proof}
First the decision variables $A_{i}$, $R_{i}$ and $H_{i}^{k}$ are binary, which make the problem non-convex. 
Second, $\bar{T}_{i}^{k}$ (cf.(11)) has term $\pi_{i,1}^{k}\pi_{i,2}^{k}$, which is not jointly convex in $\pi_{i,1}^{k}$ and $\pi_{i,2}^{k}$.
\end{proof}

\color{black}

Hence, standard convex optimization solvers such as CVX, MOSEK or integer linear programming optimization solvers such as CPLEX cannot be used.

\textcolor{black}{Problem (P1) is the first stage problem and (P2) is the second stage problem. Note that while solving the first stage problem (P1) the CSP needs to consider the second stage parameters-- the access bids, the arrival rates, and the latency requirements. This is because the second stage optimal decision (and thus, the optimal profit) inherently depends on the first stage decisions. Thus, the CSP needs to consider the second stage decision while taking the first stage optimal decisions. We consider a scenario based approach where we optimize the expected profit over all the scenarios and obtain the first stage decision variables. In the second stage, the CSP optimizes (P2) for a specific scenario which has been realized. Note that while solving the second stage problem, the first stage decisions are known.}

\section{Solution Methodology}\label{sec:algo}
\textcolor{black}{\subsection{Discussion of Computational Complexity}}
\textcolor{black}{We now demonstrate the computational complexity of the original problem with an example. Suppose we have 10 users and 1 time period with 3 scenarios, so here we have 50 binary decision variables (10$A$, 10$R$, 10$H^{1}$, 10$H^{2}$ and 10$H^{3}$). The total number of branches in a decision tree is $(2^{10})^{5}=2^{50}$, which increases exponentially with the increase of number of users, time periods and numbers of scenarios. Thus, the problem scale is very large even with a small number of users and scenarios. In addition, the latency constraint is nonlinear, which makes our problem even harder as we cannot use MILP. }

\subsection{Integer Relaxation}
Problem (P1) and (P2) are non-convex as the variables $A_i, R_i$ and $H_i^{k}$ are binary. If we relax the binary constraints, the rest of the problem will still be non-convex as the latency function $\bar{T}_i^{k}$( cf.(\ref{eq5})) is still jointly non-convex in $\pi^{k}_{i,j}$. However, if we relax the integer constraint then the objective function and the constraints will be differentiable. We can use the solver such as CONOPT  \cite{conopt} to find a locally optimal solution. CONOPT is generally used for smooth continuous functions.  It finds the solution which satisfies the KKT conditions. If the gradient is non-linear, it will be approximated via Taylor series up to the first order term. 

However, if we relax the integer constraint, the solution may not be integer, rather a value in the interval $(0,1)$. To eliminate those solutions, we need to add a penalty function which will put high penalty ($-\infty$ for optimality) when the solution is not either $0$ or $1$ and $0$ penalty when the solution is indeed $0$ or $1$ (Fig.~\ref{fig:penalty}).

\begin{figure*}
\begin{minipage}{0.35\linewidth}
\includegraphics[width=\textwidth]{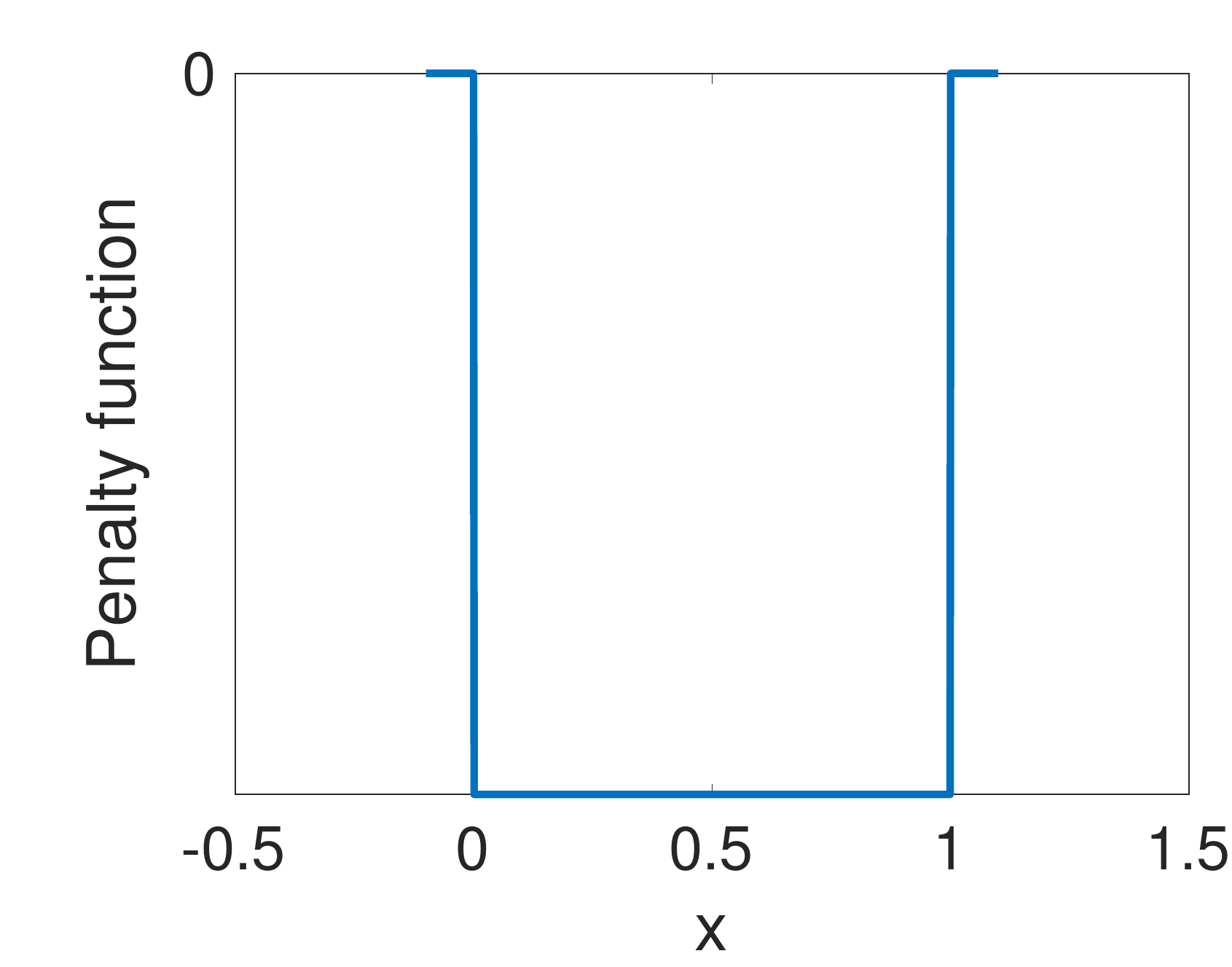}
\vspace{-0.2in}
\caption{Desired Penalty function}
\label{fig:penalty}
\vspace{-0.2in}
\end{minipage}\hfill
\begin{minipage}{0.6\linewidth}
\includegraphics[trim=0in .1in 0in .5in, clip, width=\textwidth, height=2in]{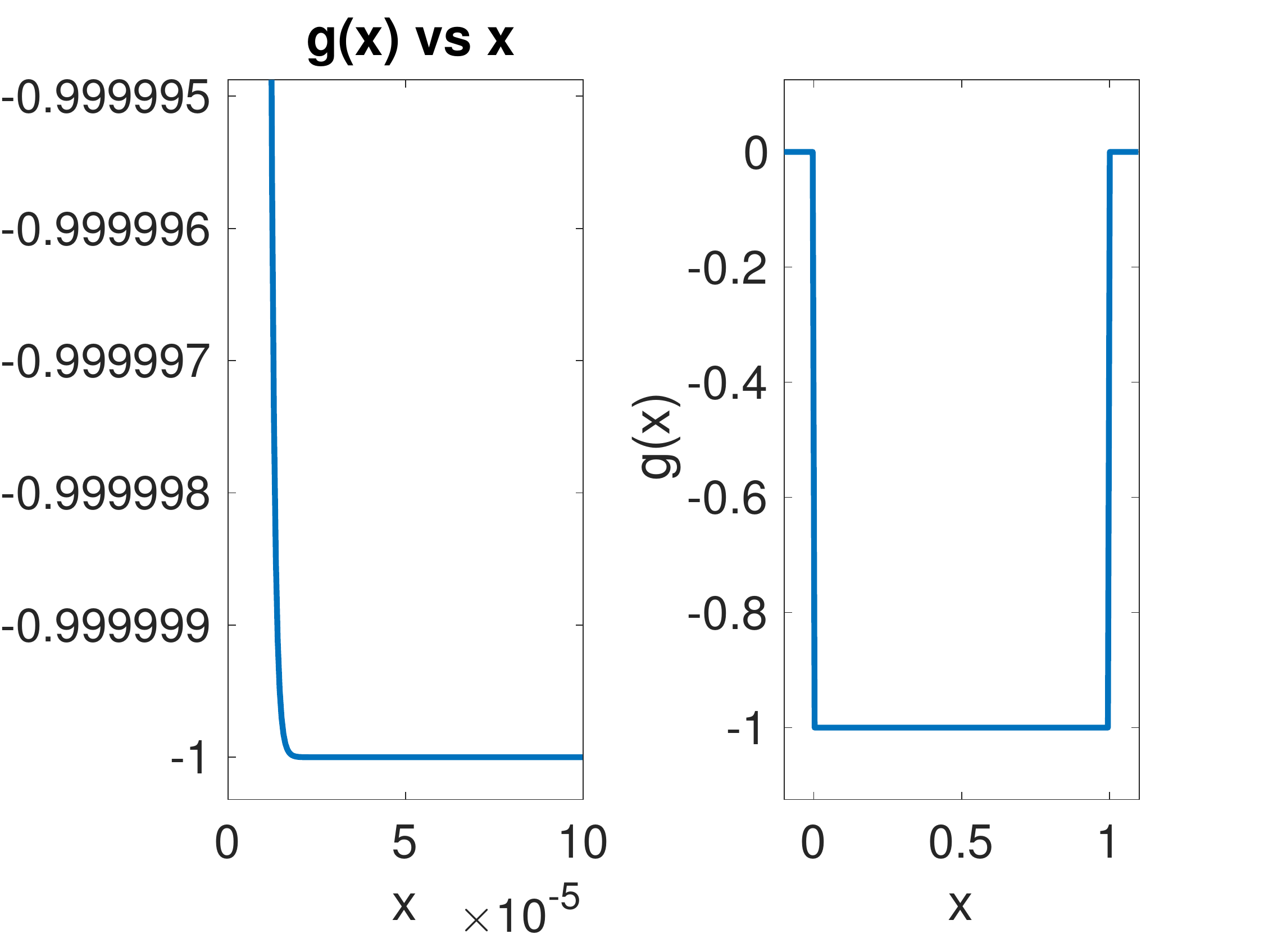}
\vspace{-0.2in}
\caption{Penalty function $g(x)$ (cf.(\ref{eq:pen})) for $\alpha=10^6$. In the left figure $x\in [0,10^{-4}]$ and in the right hand figure $x\in [-0.1,1.1]$.}
\label{fig:alpha}
\vspace{-0.2in}
\end{minipage}
\end{figure*}
Sigmoid function\footnote{$S(t)=1/(1+\exp(-t))$} is a S-shaped function which can closely approximate the step function $U(\cdot)$.   Since we have to put zero penalty when the solution is $0$ or $1$ and a high penalty when it is in between, thus, we consider the following function
\begin{align}\label{eq:pen}
g(x)=\dfrac{1}{1+\exp(\alpha x)}-\dfrac{1}{1+\exp(\alpha (x-1))}.
\end{align}

Further, let 
\begin{equation}
g_1(x)= g(x)+\frac{1}{2} - \dfrac{1}{1+\exp(\alpha)}.
\end{equation}

Fig.~\ref{fig:alpha}  shows that $g(x)$ becomes close to $-1$ at the value $\dfrac{1}{\alpha}$. Fig.~\ref{fig:alpha} also shows that  for $\alpha=10^6$, $g(x)$  closely matches the penalty function that we desire. The function $g_1(x)$ shifts the penalty function to have zero value on the desired extremes and a negative value in the desired range.  The next result shows that $g_1(0)=g_1(1)=0$. 

\begin{lemma}
The value of function $g_1(x)$ is zero for $x=0$ and $x=1$, or 	$g_1(0)=g_1(1) = 0$. 
\end{lemma}
\begin{proof}
	\begin{eqnarray}
	g_1(0)&=& \dfrac{1}{1+1}-\dfrac{1}{1+\exp(-\alpha)} + \frac{1}{2}  - \dfrac{1}{1+\exp(\alpha)}\\
	&=& \dfrac{\exp(-\alpha)}{1+\exp(-\alpha)} - \dfrac{1}{1+\exp(\alpha)}\\
	&=&0.
	\end{eqnarray}
	
Similarly, 
	\begin{eqnarray}
g_1(1)&=& \dfrac{1}{1+\exp(\alpha)}-\dfrac{1}{1+1} + \frac{1}{2} - \dfrac{1}{1+\exp(\alpha)}\\
&=& 0.
\end{eqnarray}

\end{proof}

Thus, we note that as  the  function $g_1(x)$  gives $0$ penalty when $x$ is $0$ or $1$. For any $0<x<1$, $g_1(x)\to -1/2$ as $\alpha\rightarrow \infty$. Further, even for finite $\alpha$, we note that $g_1(x)<0$ for $0<x<1$. Thus, for large $\alpha$, this function it will match the ideal penalty function. Note we do not have to lose any differentiability property as $g(\cdot)$ is differentiable. With this penalty function our problem reduces as follows.
\begin{align}
& {\text{(P3) }} \nonumber\\
& {\text{ max}}
& & \sum\limits_{i} P_{i}A_{i}-\sum\limits_{i}S_{i}(2A_{i}-R_{i})c_{1}-\sum\limits_{i}S_{i}R_{i}c_{2}\nonumber    \\
& &  & +  T \sum\limits_{i}\sum\limits_{k} p^{k} q_{i}^{k}H_{i}^{k} +\sum\limits_{i} C(g_1(A_i)+g_1(R_i)) \nonumber \\
& & &+\sum\limits_{i}\sum\limits_{k}Cg_1(H^{k}_i)
\label{P4} \\
& \text{s.t. } & &  (\ref{green-constraint-1})-(\ref{green-constraint-3}), (\ref{P2-2})-(\ref{P2-8})\nonumber\\
& & & 0\leq A_i\leq 1, 0\leq R_i\leq 1\nonumber\\
& & & 0\leq H_i^{k}\leq 1.
\end{align}

$C$ is the weight corresponding to the penalty functions. Note that the solution will be integer if $C\rightarrow \infty$.  
$A_i$ and $R_i$ are decided by solving (P3). After $A_i$ and $R_i$ are solved for a given realization $k$, the second stage decisions are taken. In the second stage, the following optimization problem is solved
\begin{align}
& {\text{(P4) }} \nonumber\\
& {\text{ max}}
& & \sum\limits_{i}(q_i^{k}H_i^{k}+Cg_1(H_i^{k}))\nonumber   \\
& \text{s.t. }
& &  (\ref{green-constraint-8}), (\ref{green-constraint-11}), (\ref{green-constraint-9}),(\ref{green-constraint-7}),(\ref{eq5}), (\ref{green-constraint-6})\nonumber\\
& & & 0\leq H_i^{k}\leq 1.
\end{align}

The decision variables are $H_i^{k}$ and $\pi_{i,j}^{k}$. Note that we do not need to decide $A_i$ and $R_i$ in the second stage, hence, we do need constraints (\ref{green-constraint-1})-(\ref{green-constraint-3}). Note that the solution $H_i^{k}$ will be optimal and integer if $C\rightarrow \infty$. 
Since the problem is non-convex, we cannot guarantee that the solution obtained by CONOPT will be optimal. However, we can infer the following if we find an optimal solution
\begin{proposition}\label{prop:optimal}
The optimal solution of the relaxed problem (i.e. (P3), (P4)) is also the optimal solution of the original problem (i.e. (P1), (P2)) as $C\rightarrow \infty$.
\end{proposition}

\subsection{Feasible solution from the relaxed problem}
When $C$ is $\infty$, both the first-stage and second-stage decision solutions $A_i, R_i$, and $H_i^{k}$  will be integers.   If $\alpha\rightarrow \infty$ $g(x)$ will match the ideal penalty function. However, in practice, neither $C$ nor $\alpha$ can be set at $\infty$. Hence, we may find a solution which is not feasible, {\it i.e.} it is not either $0$ or $1$. Note that setting $\alpha$ alone to the very high value will not make the solution integer. One also has to make $C$ high to give larger penalty to the fractional solution.  However, $C$ has to be larger for smaller $\alpha$. In the following, we discuss how to find the feasible solution for  finite $C$ and $\alpha$. 

Also note that if $C$ is very high, in an optimal solution the solution will only be away from the integral solution by a nominal amount. One can then convert the non-integral solution of either $H^{k}_i, A_i, R_i$ to the nearest integer. However, the above does not guarantee that the capacity constraints or the latency requirements will be satisfied.  For example, consider that  in a solution of the relaxed problem $A_i=1-\epsilon$, where $\epsilon>0$ is very small, and $A^{r}_i=1$ is the nearest integer solution to the relaxed problem. However, if $\sum_{i}(2A_i-R_i)S_i=C_1$, then $\sum_{i}(2A_i^{r}-R_i)S_i>C_1$ which violates the constraint in (\ref{green-constraint-1}).  
Thus, simple converting the solution $A_i, R_i, H_i^{k}$ of the relaxed problem to the nearest integer may not give a feasible solution. However, in the following, we provide a strategy which can guarantee that even if the solution of the relaxed problem is converted to the nearest integer, then, it will not violate the original constraint. 
\begin{proposition}
For every $C$ and $\alpha$, there exists an $1>\epsilon>0$ such that if $C_j=C_j(1-\epsilon)$ and $\mu_j=\mu_j(1-\epsilon)$, such that if the solution $A_i, R_i, H^{k}_i$ of the relaxed problem ({\it i.e.}, (P3), (P4)) is converted to the nearest integer (if the value is $0.5$, it will be converted to $0$) then they will be feasible solution of the original problem ({\it i.e.}, (P1), (P2)).
\end{proposition}
Intuitively, if we make $C_j=C_j(1-\epsilon)$ and $\mu_j=\mu_j(1-\epsilon)$, we solve a restricted problem. Thus, even when we convert the non-integer solutions of the relaxed problem to the nearest integers we will not violate the original constraints.  
  Note that if $C$ is very large, we need a very small $\epsilon$ as the solutions of the relaxed problem $A_i, R_i$ and $H^{k}_i$ will be close to the integers. As $C\rightarrow \infty$, $\epsilon\rightarrow 0$.  $C$ is also larger if $\alpha$ is low. In our numerical results, we set $\epsilon$ as $0.001$, $\alpha$ as $10^6$, and $C$ as $10^9$ which gives the feasible solutions as mentioned in the above proposition.

\section{Numerical Studies}\label{sec:simulation}

\subsection{Simulation Setting}
To validate our proposed policy and evaluate its performance, we implement the following numerical studies. Unless stated otherwise, we consider a setting where there are 1,000 files,  and the number of slots for the second-stage auction is $T=20$.  
%
The capacities of cold and hot storage are 400 and 200 GB respectively. We consider five types of files: $Ty^{\Rmnum {1}}$, $Ty^{\Rmnum {2}}$, $Ty^{\Rmnum {3}}$, $Ty^{\Rmnum {4}}$ and $Ty^{\Rmnum {5}}$,  which are  of sizes  64, 128, 256, 512 and 1024 MB respectively.  In the first stage, customers will bid for storage. Bidding prices for storage per MB are considered to be a random variable $\sim \mathcal{U}[0.1,0.3]$ i.e., it is uniformly distributed with a mean of 0.2 cents.  Thus, the  bid $P_{i}$ for file $i$ is distributed as $\sim S_{i}*\mathcal{U}[0.1, 0.3]$. For example, if there is a 64 MB file and the realized price is 0.25 cents for each MB\footnote{Note that in Amazon they put $2.5$\$ for each GB.}, the bidding price to store this file is $64*0.25=16$ cents.

We consider $K=10$ different scenarios for the second-stage parameter. Specifically, we generate $10$ different instances of access request arrival rates, access bids, and the latency requirements. We consider that $\lambda_i^{k}$ is generated independently according to the mean $20, 10, 8, 4,$ and $2$ per hour for the file sizes of  64, 128, 256, 512 and 1024 MB respectively. This is in accordance with the practice as the smaller size files are accessed more frequently.  We assume that the latency requirements $l_{i}^{k}$ are related to the file sizes. Specifically, we generate $l_{i}^{k}$ independently according to the distribution $\sim \mathcal{U}[30+\frac{S_{i}}{5*10^6},30+\frac{S_{i}}{10^6}]$ in milliseconds. \cite{lz04} shows that the utility is in general convex in the latency and concave in the arrival rates.  In this paper, we consider that $q_{i}^{k}=\frac{50S_{i} \log(\lambda_{i}^{k}+1)}{(l_{i}^{k})^{2}} $. The parameters are described in Table~\ref{my-label}. After generating the $K$ scenarios we compute the empirical distribution to find the number of times a scenario $k$ (prob. $p^{k}$) occurs out of the $10$ events.  Then scenario is randomly generated among the $K$ scenarios where $k$-th scenario occurs with probability $p^{k}$.


\begin{table*}[ht]
\centering
\caption{Bidding Scenario $k$ Generation}
\label{my-label}
\begin{tabular}{@{}cclccc@{}}
\toprule
               & \begin{tabular}[c]{@{}c@{}}File Size\\ $S_{i}$\end{tabular} & \begin{tabular}[c]{@{}l@{}}Bidding Price \\ for Storage\\ $P_{i}$\end{tabular} & \begin{tabular}[c]{@{}c@{}}Mean of \\ Arrival Rate \\ $\lambda_{i}^{k}$\end{tabular} & \begin{tabular}[c]{@{}c@{}}Latency \\ Requirements\\  $l_{i}^{k}$\end{tabular}&\begin{tabular}[c]{@{}c@{}}Bidding Price \\ for Access \\ $q_{i}^{k}$\end{tabular}                                     \\  \midrule
$Ty^{\Rmnum{1}}$ & 64        & $64*\mathcal{U}[0.1,0.3]$           & 20                                     & $\mathcal{U}[30+\frac{64}{5*10^6},30+\frac{64}{10^6}]$              & $\frac{50*64* \log(\lambda_{i}^{k}+1)}{(l_{i}^{k})^{2}}$   \\
$Ty^{\Rmnum{2}}$ & 128       & $128*\mathcal{U}[0.1,0.3]$          & 10                                     & $\mathcal{U}[30+\frac{128}{5*10^6},30+\frac{128}{10^6}]$              & $\frac{50*128* \log(\lambda_{i}^{k}+1)}{(l_{i}^{k})^{2}}$  \\
$Ty^{\Rmnum{3}}$ & 256       & $256*\mathcal{U}[0.1,0.3]$          & 8                                      & $\mathcal{U}[30+\frac{256}{5*10^6},30+\frac{256}{10^6}]$             & $\frac{50*256* \log(\lambda_{i}^{k}+1)}{(l_{i}^{k})^{2}}$  \\
$Ty^{\Rmnum{4}}$ & 512       & $512*\mathcal{U}[0.1,0.3]$          & 4                                      & $\mathcal{U}[30+\frac{512}{5*10^6},30+\frac{512}{10^6}]$             & $\frac{50*512* \log(\lambda_{i}^{k}+1)}{(l_{i}^{k})^{2}}$  \\
$Ty^{\Rmnum{5}}$ & 1024      & $1024*\mathcal{U}[0.1,0.3]$         & 2                                      & $\mathcal{U}[30+\frac{1024}{5*10^6},30+\frac{1024}{10^6}]$            & $\frac{50*1024*\log(\lambda_{i}^{k}+1)}{(l_{i}^{k})^{2}}$ \\ \bottomrule
\end{tabular}
\end{table*}

Based on the above specifications, we compare the performances of the proposed method (PM) with three other  methods, which are described as follows.  We consider $\alpha$ as $10^6$ and $C$ as $10^9$ in (P3) and (P4). The factor by which we reduce $C_j$ and  $\mu_j$ is choisen to be $\epsilon=0.001$. The solution obtained by the relaxed problem and the proposed method are almost the same. Thus, we do not show the solution of the relaxed problem. 

\begin{itemize}
\item IS: Problem with Two Independent Stages:
\begin{itemize}
\item Solve the first-stage  problem without considering the second-stage recourse decisions to get the first-stage solution $A_i$ and $R_i$ for each $i$.
\item Given the first-stage solution solve for the realized scenario, {\it i.e.} solve the second stage optimization problem (P2). We again solve the relaxed version (P4) and then find the optimal solution according to Proposition 4.2 as described in our proposed method

\end{itemize}
\item GH I: Greedy Heuristic Based On $q_{i}^{k}/S_{i}$:
\begin{itemize}
\item Solve the first-stage  problem without considering the second-stage recourse decisions to get the first-stage solution $A_i$ and $R_i$ for each $i$.
\item In the second stage, we sort the bids based on $q_{i}^{k}/S_{i}$ in the descending order. We keep accepting bids according to the sorted order as long as the realized the latency requirements are met.
\end{itemize}
\item GH II: Greedy Heuristic Based On $q_{i}^{k}/ \lambda_{i}^{k}$:
\begin{itemize}
\item Solve the first-stage  problem without considering the second-stage recourse decisions to get the first-stage solution $A_i$ and $R_i$ for each $i$.
\item In the second stage, we sort the bids based on $q_{i}^{k}/\lambda_i^{k}$ in the descending order if scenario $k$ is realized. We keep accepting bids according to the sorted order as long as the realized the latency requirements are met.
\end{itemize}
\end{itemize}
Profit in each algorithm is considered to be the sum of the first-stage and second-stage profits. 
Note that all the above mentioned base-line algorithms do not solve the first-stage decision problem by considering the second-stage recourse decision. Algorithm IS solves the second-stage decision problem given the solution of the first-stage decision. However, GH I and GH II are greedy heuristics which accept bids according to some heuristics in order to lower the complexity of finding the optimal solution of the second-stage decision problem. Intuitively, recall from (\ref{eq5}) that the latency of a file  in scenario $k$ inherently depends on the access request arrival rates and file size. Specifically, the latency increases as the file size increases or the access request arrival rate increases. Hence, the CSP should prefer the bids which give more profit per unit of the size and the per unit of the access request arrival rate. GH I greedily prefer the bids which pay more per unit of size. On the other hand, GH II strictly prefers the bid which pays more for per unit of access request rate.  
Before discussing the results, we introduce a notation which we use throughout this section.
\begin{align}  \label{ass}
 &Access Request Acceptance Rate (ARAR) \nonumber\\ 
&= \frac{\text{Total Number Of Accessed Files}  }{ \text{ Total Number of Requests} }
\end{align}

The above metric shows how much bids are accepted in the second stage among the bids that are accepted in the first stage. This will give an idea pertaining the fairness of the process. 
In each of the result,  each algorithm is run $100$ times and an average is taken for the  profit and ARAR over these runs.
\subsection{Impact of Storage Capacity}

\begin{figure}[ht]
\centering

\subfigure[Impact of the variation of cold storage capacity on profits]{
   \includegraphics[scale =0.21] {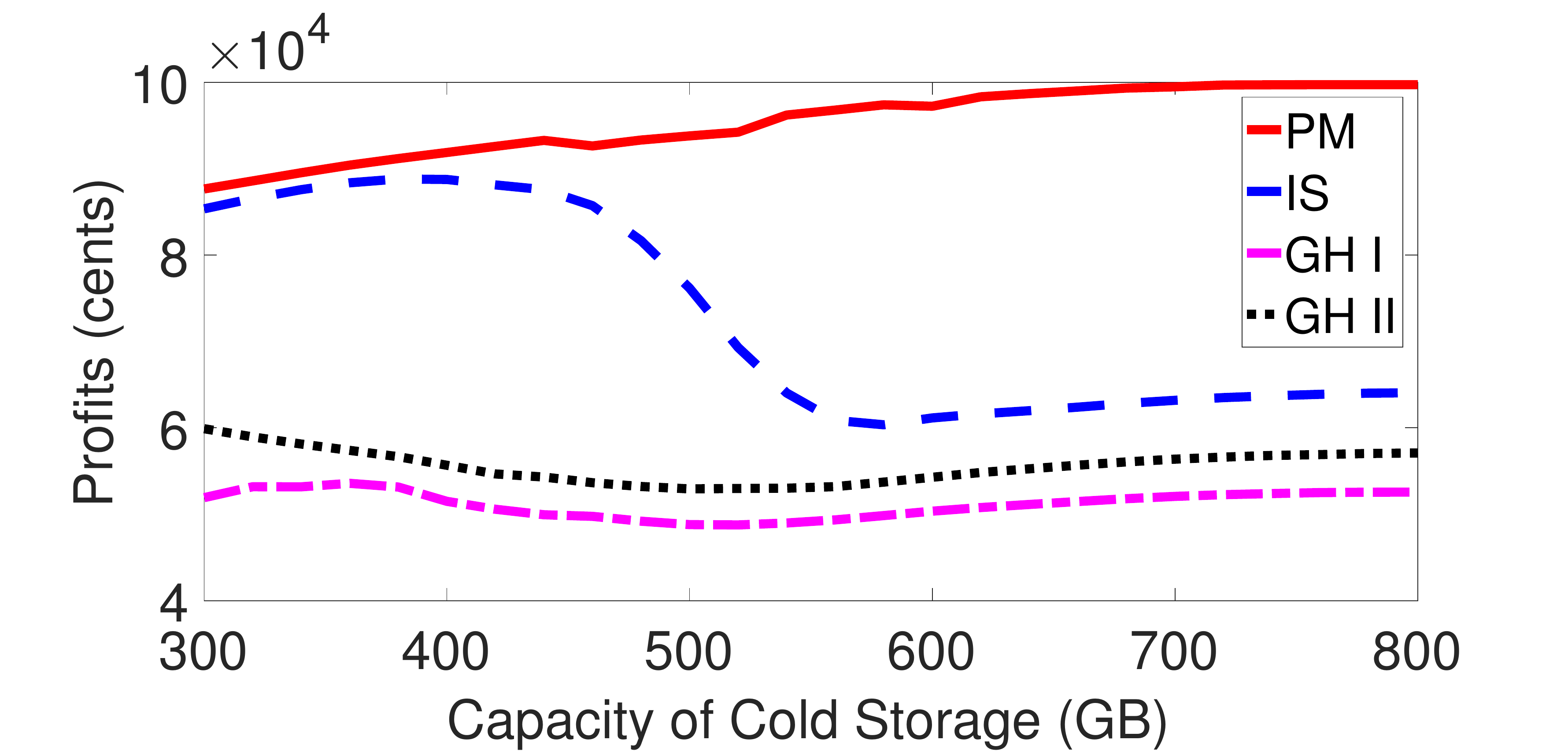}
   \label{fig:1.1}
 }
 \subfigure[Impact of the variation of cold storage capacity on profits obtained from file storage and access in PM]{
   \includegraphics[scale =0.21] {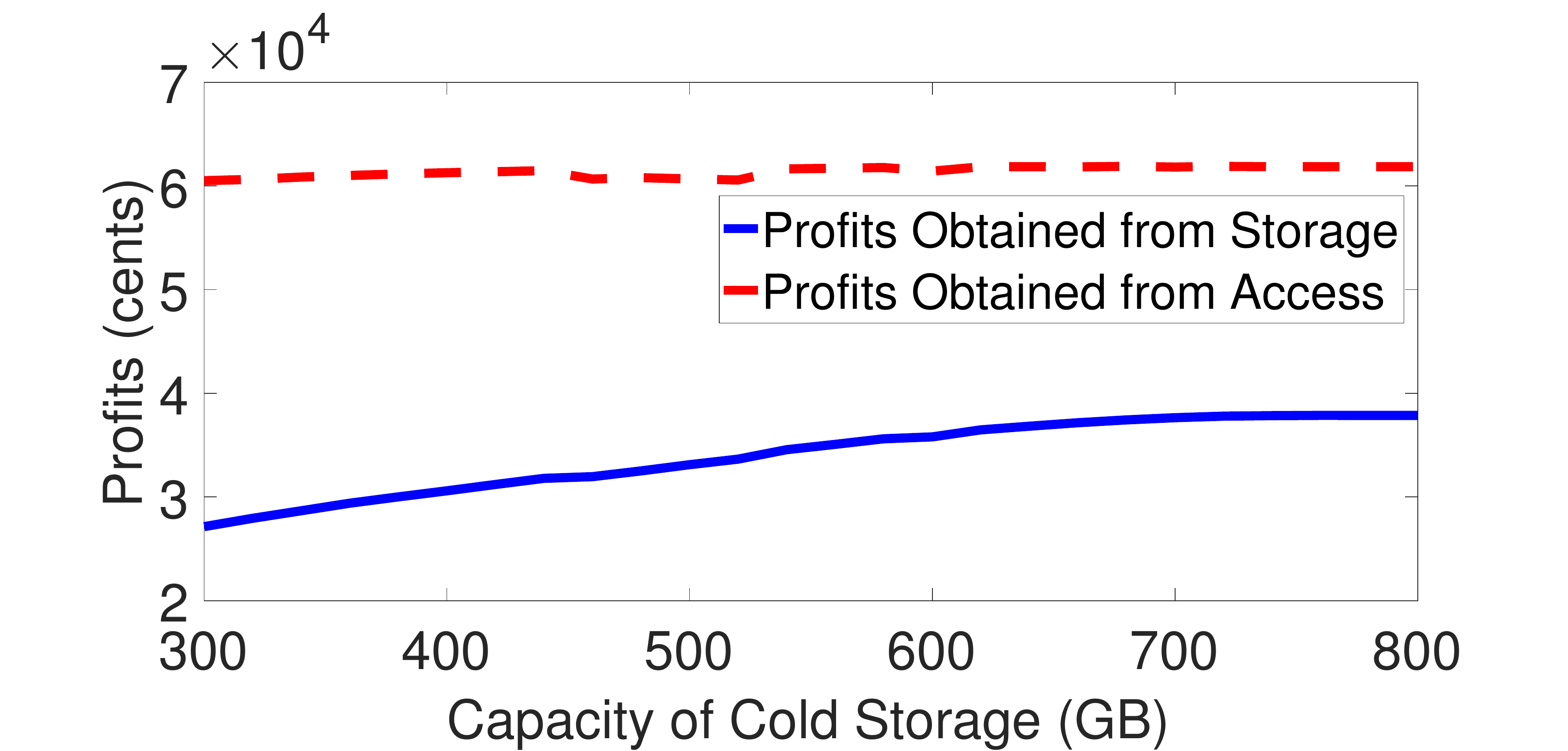}
   \label{fig:1.2}
 }
 \subfigure[Impact of the cold Storage capacity on ARAR]{
   \includegraphics[scale =0.21] {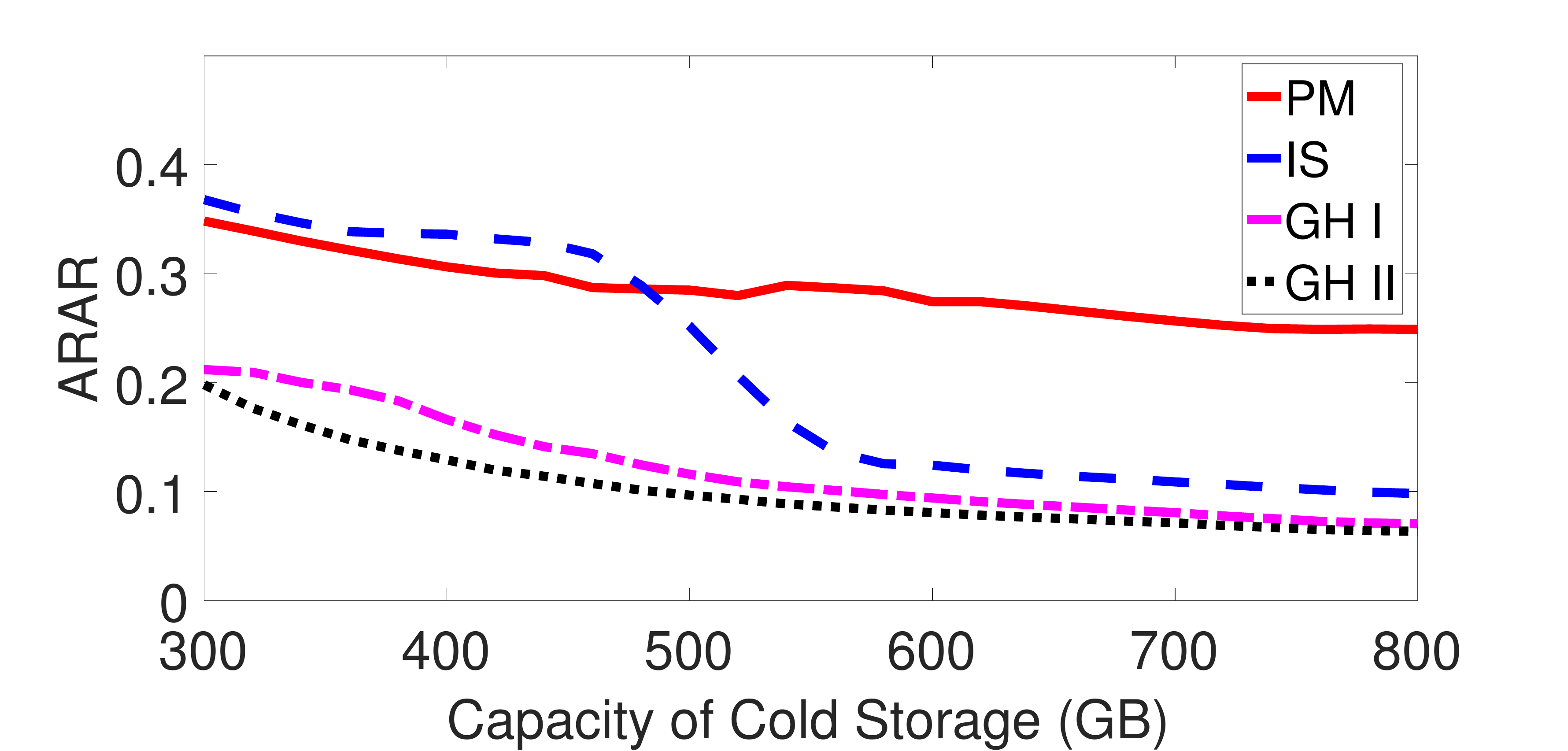}
   \label{fig:1.3}
 }
 \subfigure[Impact of the cold storage capacity on the number of files that are accepted for storage and number of bids accepted for access]{
   \includegraphics[scale =0.21] {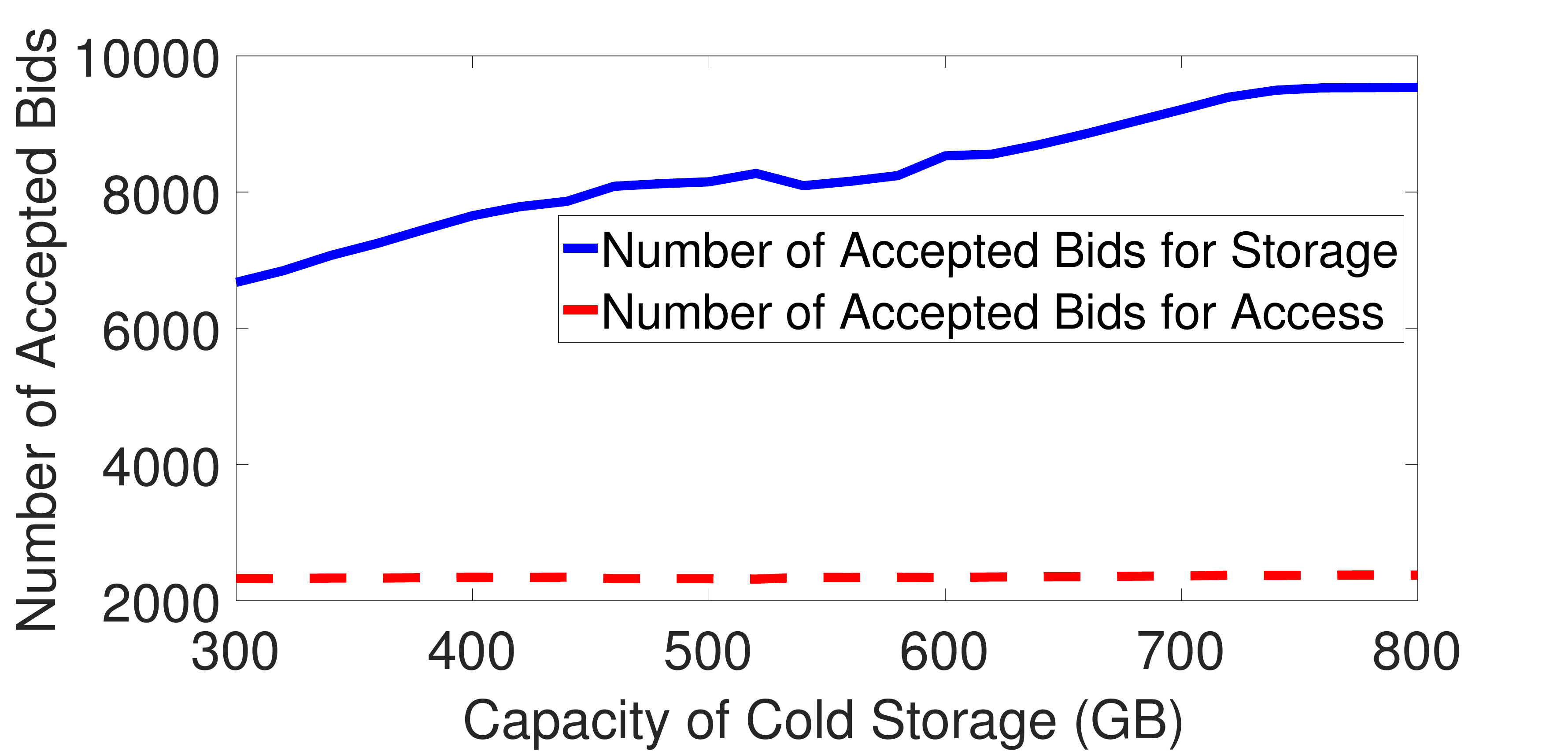}
   \label{fig:1.4}
 }
\caption{\small The impact of the variation of Cold storage capacity. The cost of cold storage and hot storage is $50$ cents per GB and $80$ cents per GB respectively. The service rates of the cold and hot storage are $100$Gb/s and $200$Gb/s. }
\vspace{-0.25in}
\label{fig1}
\end{figure}

 To demonstrate the effectiveness of our proposed heuristic, we fix the hot storage capacity as 200 GB and vary the capacity of cold storage ($C_1$) from 300 GB to 800 GB in the steps of 20 GB, and plot the total profits and access request acceptance rate (cf. (\ref{ass}))  by using different methods. Fig. \ref{fig:1.1} shows that as the capacity of cold storage ( or ratio $C_{1}/C_{2}$ ) increases, the profits obtained from all the algorithms except IS increases; however, the rate of increase decreases with the increase in the $C_1$.  Fig.~\ref{fig:1.4} provides the reason behind this variation. As $C_{1}$ increases, more files can be stored in cold storage which increases profits. However, if $C_1$ is large enough, no more files can be stored, thus, the profit becomes saturated. Fig.~\ref{fig:1.3} shows that because of the lower service rate of the cold storage, the profit from accessing the file does  not increase with the increase in the capacity of cold storage. This is because files may be stored but cannot be accessed as it violates the latency constraint.   {\em Thus, increasing the cold storage capacity without increasing the hot storage capacity will not fetch more profit after a certain threshold.} Note from Fig. \ref{fig:1.1} that the profit achieved by Algorithm IS increases initially, then decreases and again increases as $C_1$ increases. Intuitively, the Algorithm IS does not consider the second stage decision variables in its first-stage decision. Hence, more files are stored in the cold storage as it has a lower cost. However, as almost all the files are stored in the cold storage, the files cannot be accessed fast enough which does not increase the profit from accepting the access bids. Similarly, the profits earned by Algorithms GH I and GH II do not increase much as $C_1$ increases as they do not consider the second-stage recourse decisions in the first stage.  
  Also note that when $C_1$ is large, our algorithm outperforms the other base-line algorithms by 50\%. This shows the virtue of the consideration of the second stage recourse decision in the first-stage decision. 

From the results in Fig. \ref{fig:1.3}, the Access Req. Acceptance Rate (ARAR, cf.(\ref{ass}))  decreases as the capacity of cold storage increases. This is because, by increasing the capacity of cold storage, the number of files accepted for storage increase, however with limited cold storage service rate, the number of files accepted for access is limited (which is also verified by Fig.~\ref{fig:1.4}). Consequently, the ARAR decreases. Note that the ARAR corresponding to Algorithm IS is higher compared to our proposed method when $C_1$ is low as vary number of files are stored by the IS compared to our proposed method.

\subsection{Impact of Service Rate of Hot Storage}

\begin{figure}[h]
\centering

\subfigure[Impact of the hot storage service rate on profits]{
   \includegraphics[scale =0.21] {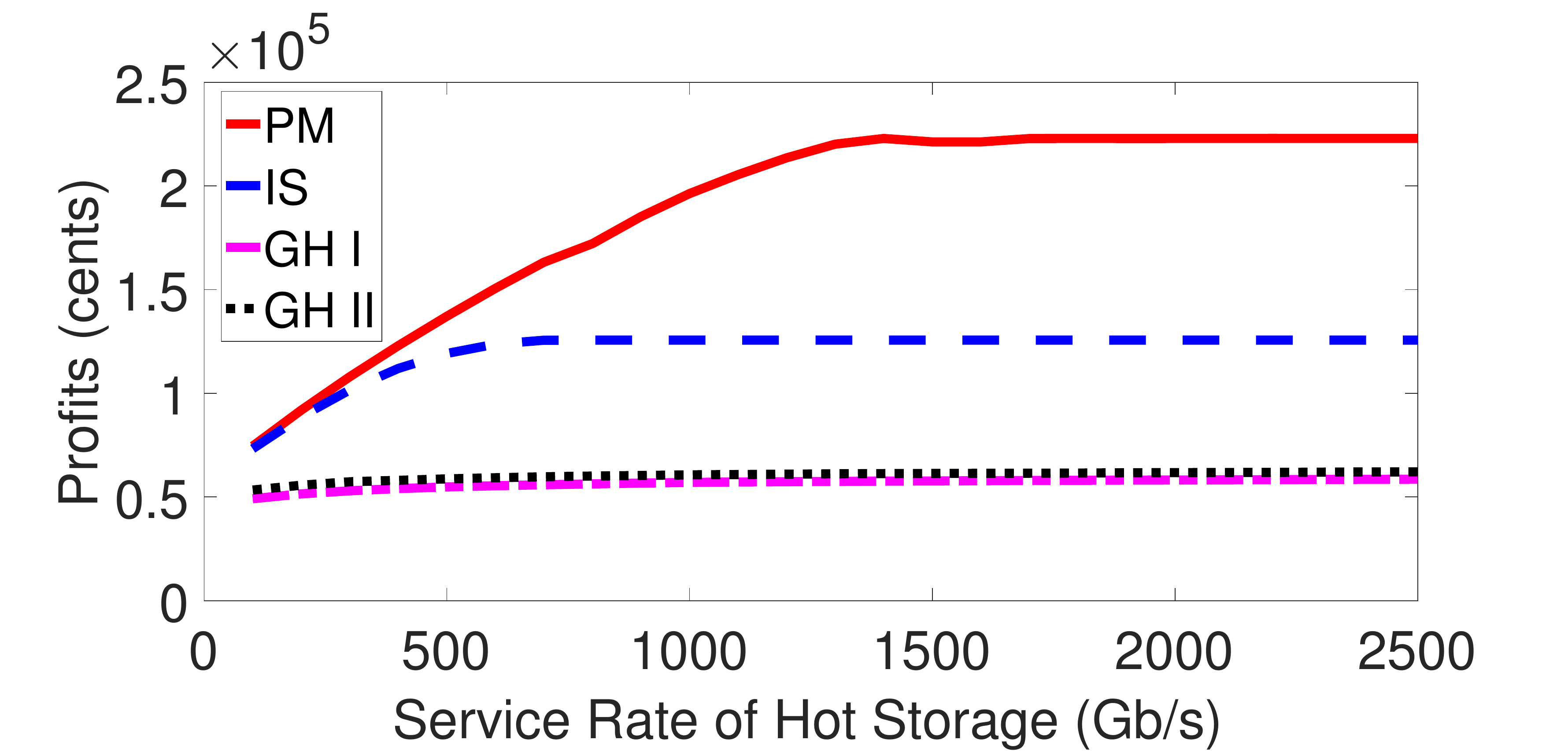}
   \label{fig:2.1}
 }
\subfigure[Impact of hot storage service rate on profits obtained from file storage and access]{
   \includegraphics[scale =0.21] {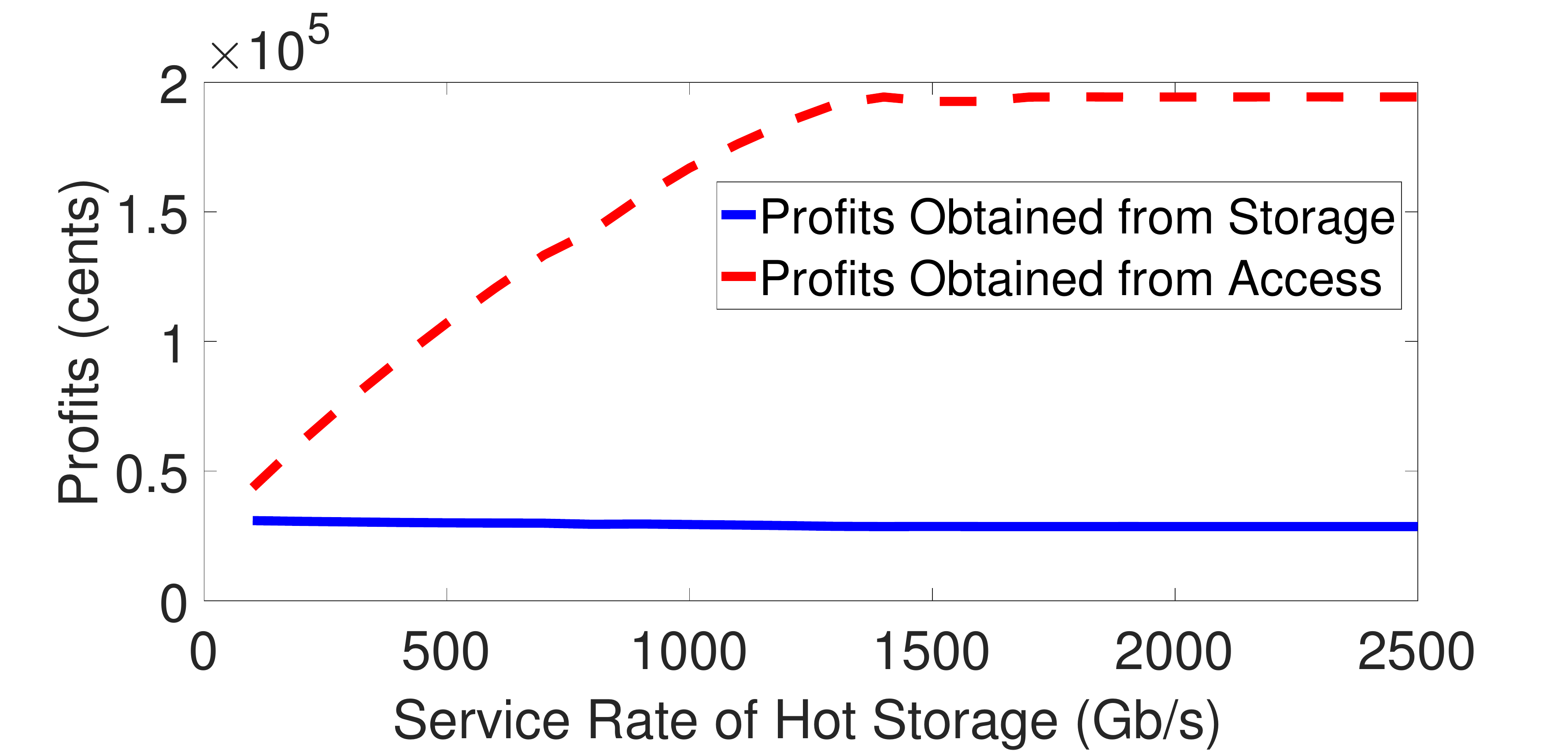}
   \label{fig:2.2}
 }
 \subfigure[Impact of hot storage service rate on ARAR]{
   \includegraphics[scale =0.21] {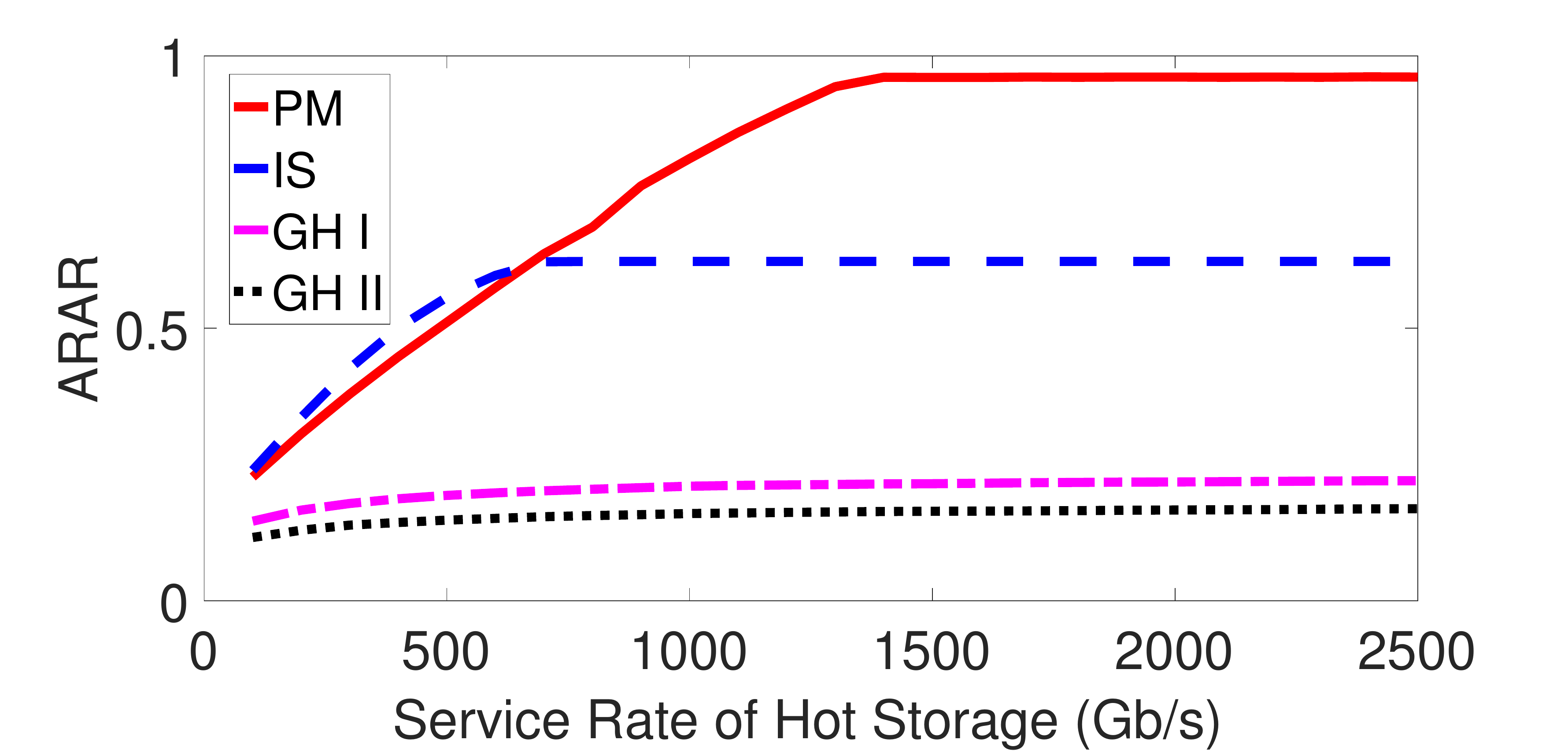}
   \label{fig:2.3}
 }
 \subfigure[Impact of  hot storage service rate on the number of bids accepted for storage and access ]{
   \includegraphics[scale =0.21] {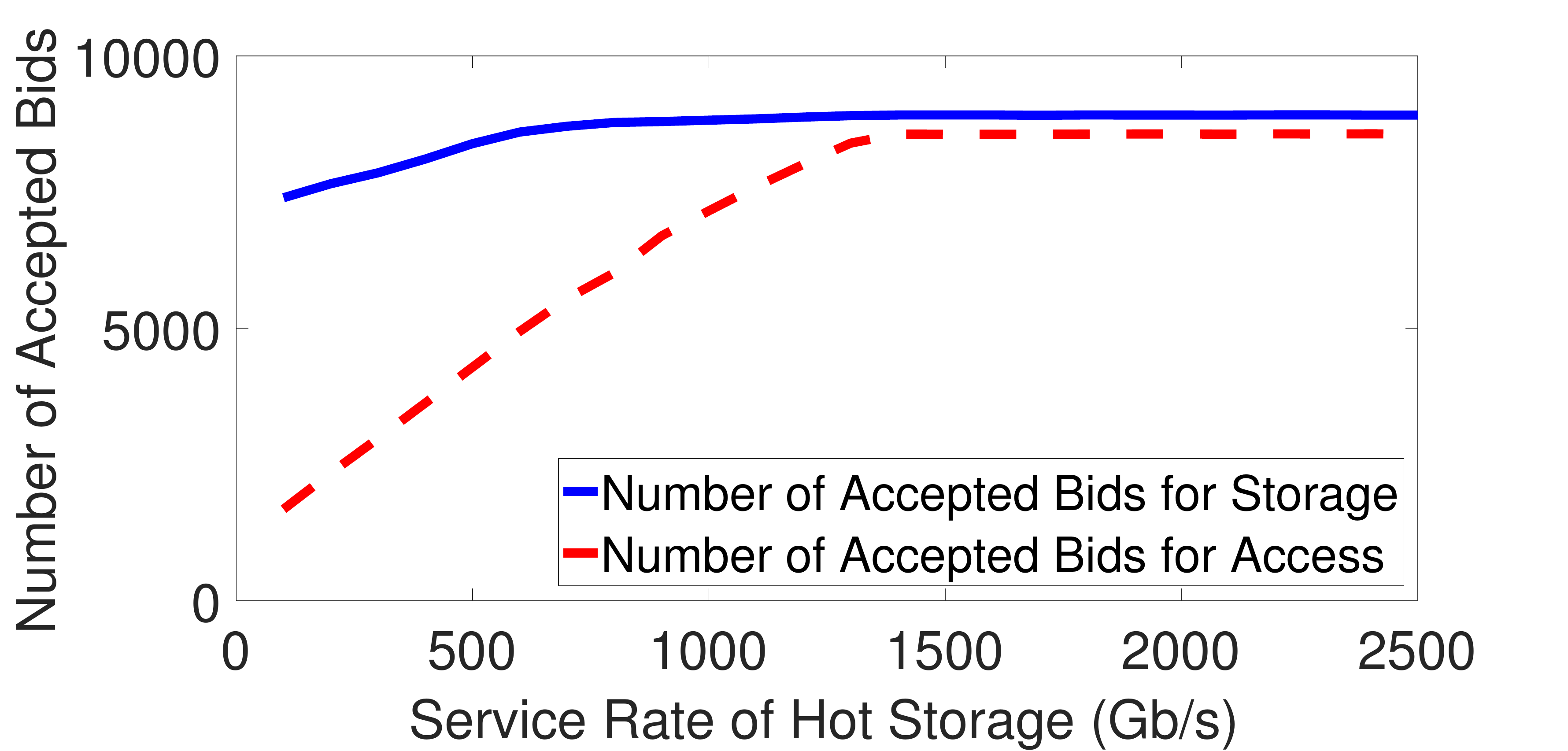}
   \label{fig:2.4}
 }
\caption{\small Profits and ARAR as a function of hot storage service rate. The cost of the cold and hot storage are $50$ cents and $80$ cents per GB respectively. $C_1=400$ GB and $C_2=200$GB. }
\vspace{-0.1in}
\label{fig2}
\end{figure}

In this subsection, we assume that the  service rate of cold storage is   100Gb/s, and the service rate of hot storage is varied from 100 to 2500 Gb/s in steps of 100Gb/s. Fig. \ref{fig:2.1} shows that the profit increases as the service rate of hot storage increases. This is because more access requests can be accepted as the number of accepted bids increase (Fig.~\ref{fig:2.4}). Our proposed method outperforms the other methods. In fact, the profit can be increased by 100\% compared to the other methods for high service rate. 

Note from Fig.~\ref{fig:2.2} that the profit from the second-stage auction increases significantly in our proposed method. However, the profit from the storing files does not increase much. This is because when the service rate is low mostly those files who have lower access requests or lower latency requirements (but can pay more) are accepted.  As Fig.~\ref{fig:2.4} suggests, when the service rate is high, more files are stored, however, the number of accepted access bids increases significantly. This suggests that when the service rate is high, the files which bid lower prices for storage, but still can pay more because of the high access rates are accepted for storing. Hence, the profit from storing the files remains constant as the storage capacities remain constant, however, the profits from accepting bids increase. 


Fig.~\ref{fig:2.3} shows that the ARAR increases with the service rate of the hot storage for all the algorithms. When the service rate is high, more  access bids are accepted, however, the bids accepted for storage remains the same (Fig.~\ref{fig:2.4}). Hence, the ARAR is high (cf.(\ref{ass})). When the service rate is low, mostly the files those have lower access requests are stored. However, the IS still can store files which have higher access rates if they pay more because it solves the two stage problem independently. Hence, the IS can achieve more ARAR in this case. However, when the service rate of the hot storage exceeds a threshold, the ARAR attained by our proposed method is the highest. The ARAR attained by the greedy heuristics GH I and GH II are strictly lower compared to our proposed method. 

\subsection{Impact of Storage Cost of Hot Storage}
\begin{figure}[ht]
\centering
\subfigure[Impact of  hot storage cost on profits]{
   \includegraphics[scale =0.21] {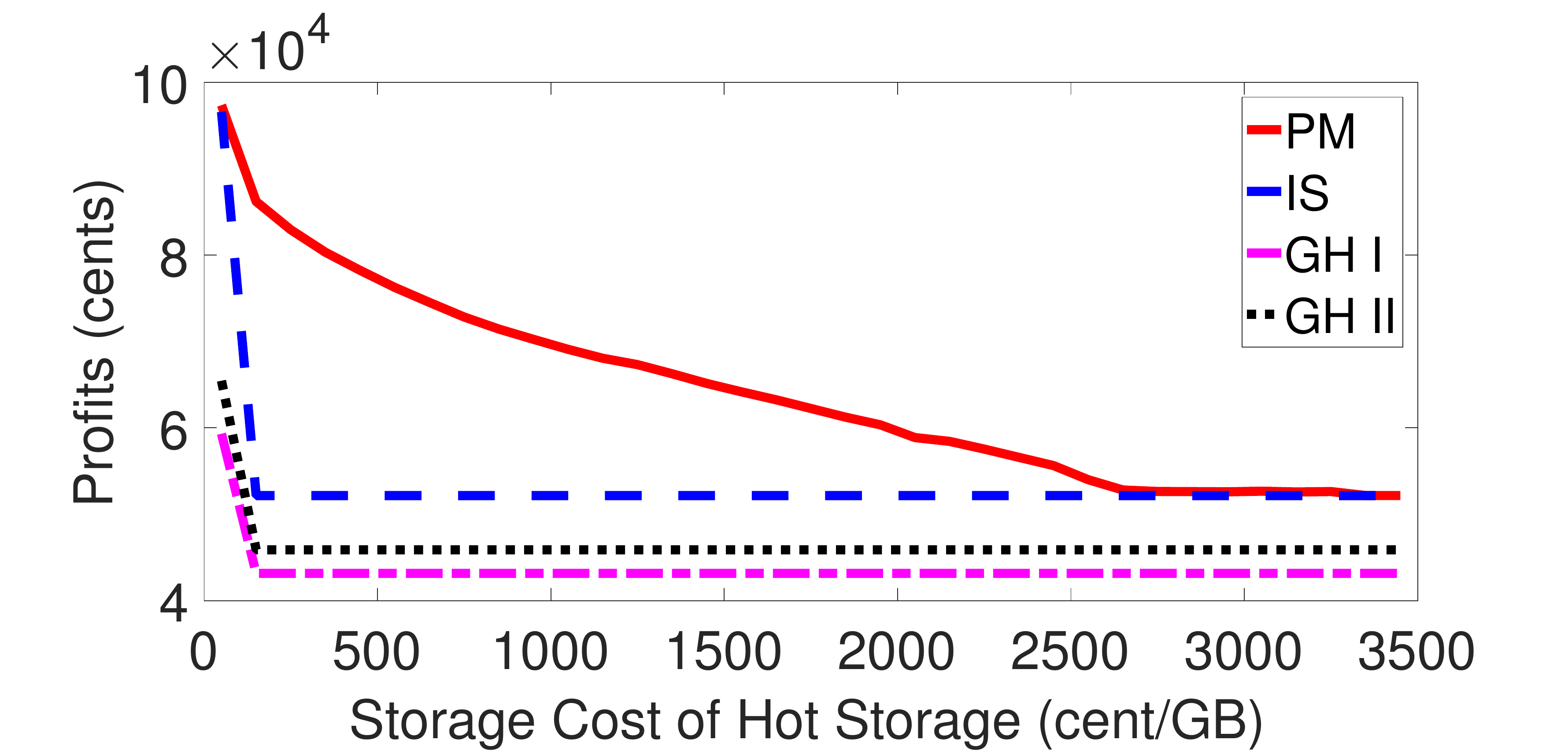}
   \label{fig:3.1}
 }
\subfigure[Impact of  hot storage cost on profits obtained from the file storage and access]{
   \includegraphics[scale =0.21] {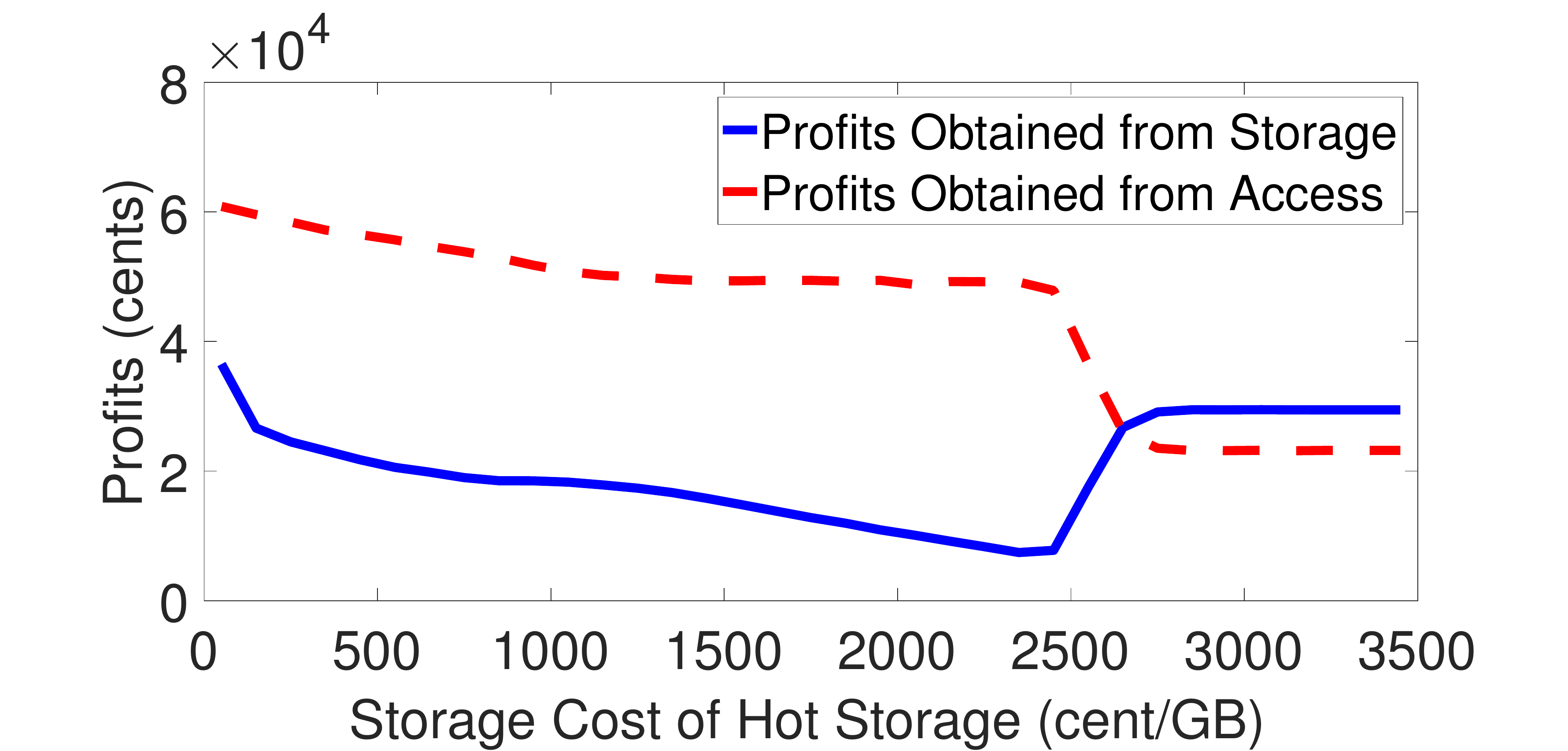}
   \label{fig:3.2}
 }
 \subfigure[Impact of hot  storage cost on ARAR]{
   \includegraphics[scale =0.21] {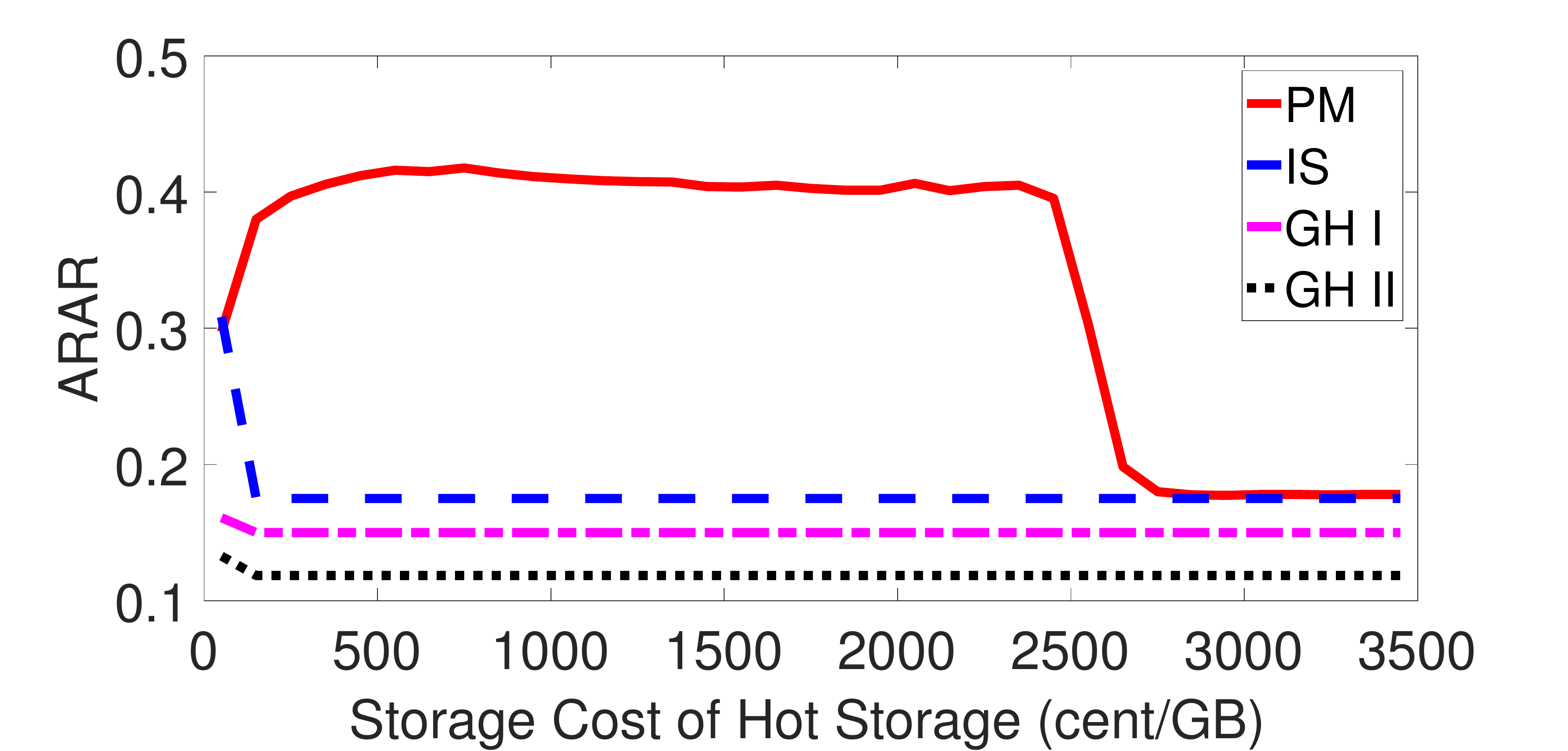}
   \label{fig:3.3}
 }
 \subfigure[Impact of hot storage cost on  the number of accepted bids for storage and access]{
   \includegraphics[scale =0.21] {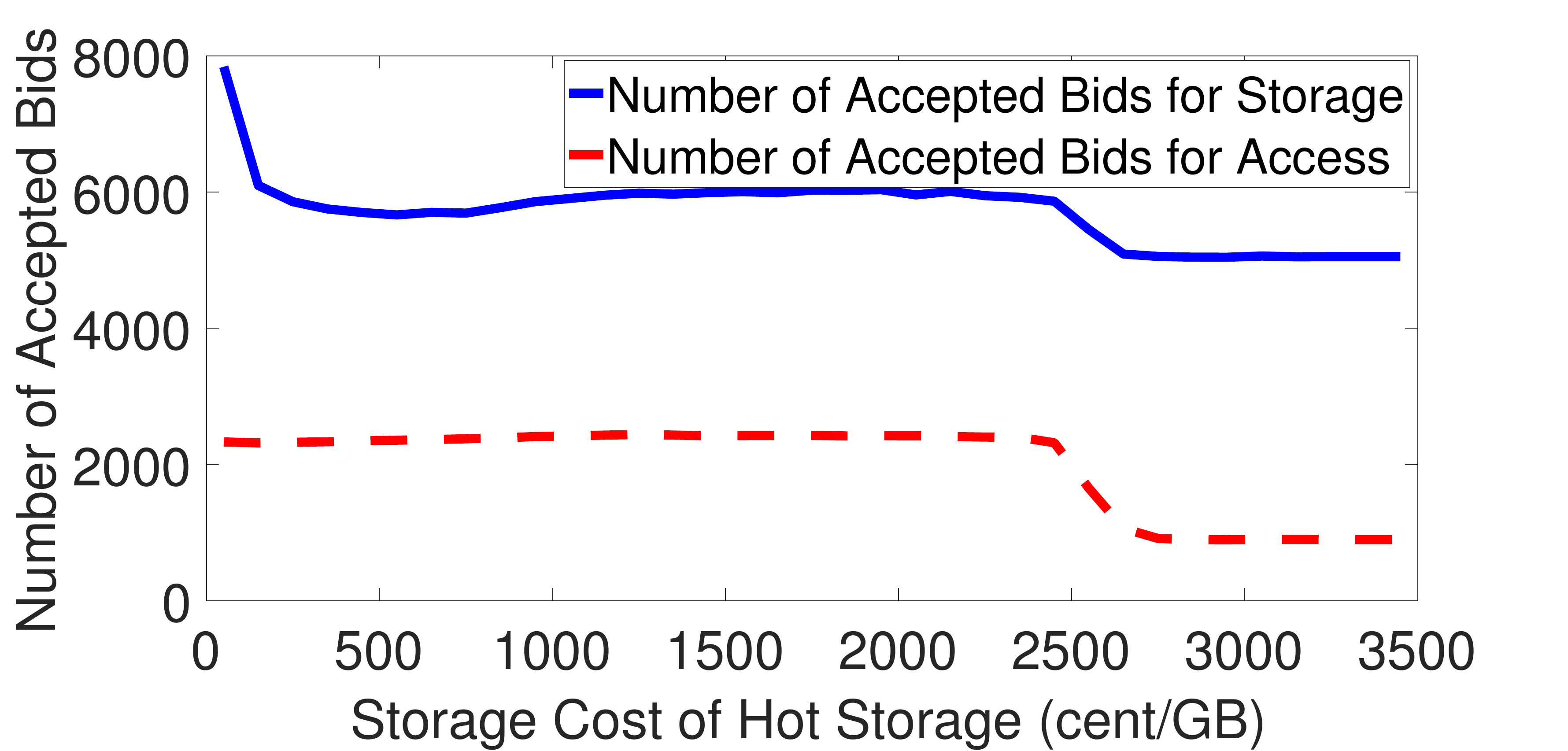}
   \label{fig:3.4}
 }
\caption{\small Profit and the ARAR as a function of the hot storage cost. $C_1=400$ GB and $C_2=200$ GB. The service rate of the cold and hot storage are respectively $100$Gb/s and $200$ Gb/s.}
\vspace{-.1in}
\label{fig3}
\end{figure}

In this subsection, we assume that the storage cost of cold storage is $50$ cents per GB, and the storage cost of hot storage is varied from $50$ to $3450$ cents per GB in the step of $100$ cents per GB.  Fig.~\ref{fig:3.1} shows that as the hot storage cost increases the profit decreases. This is because most of the files are stored in the cold storage which decreases the profit as fewer number of access requests are accepted which is also verified from Fig.~\ref{fig:3.4}.  Our proposed method outperforms the baseline algorithms by more than 60\% when the storage cost is  neither too high nor too low. When the hot storage cost is too low, more  files are stored in the hot storage in the first stage and thus, more profit can be attained in the second stage by accepting more access bids. Hence, the profit attained by IS is close to our proposed method when the hot storage cost is low. Note that the profit attained by the IS is also very close to our proposed method when the hot storage cost is high. This is because IS inherently stores more files in the cold storage in the first stage.  Since the greedy algorithms GH I and GH II do not optimize the second-stage decision, the profits attained by those are slightly lower compared to the IS.  

Note from Fig.~\ref{fig:3.2} that the profit in our proposed method from accessing the files decrease with an increase in the cost of hot storage as more files are stored in the cold storage. Though the overall profit decreases, the profit from storage increases. This is because the files which can pay more but do not have low latency requirements can be stored in the cold storage. 

We also plot the impact of  the storage cost of hot storage on the ARAR in Fig. \ref{fig:3.3}. One interesting trend is the access rate obtained from the proposed method first increases, and decrease until close to the one gained from IS, while the others decrease and then go stable with the increase in the cost. This is because we combine the two -stage decision process, thus, when the cost is moderate, the number of files that are stored decreases without decreasing the number of accepted access bids as shown in Fig.~\ref{fig:3.4}. Hence, the denominator in (\ref{ass}) decreases which increases  the ARAR. Note that once the hot storage cost is very high, the number of accepted bids also decrease as the latency requirement may not be met because too little files are stored in the expensive hot storage. Hence, the ARAR decreases at very high cost. On the other hand, in the other algorithms, as the cost of the hot storage increases, very few files are stored in the hot storage; thus, very little files can be accessed which decreases the ARAR. However, if the cost is too high, no more file can be stored in the hot storage which makes the ARAR constant.

\subsection{Impact of Storage Cost of Cold Storage}

\begin{figure}[ht]
\centering
\subfigure[Impact of  cold storage cost on profits]{
   \includegraphics[scale =0.21] {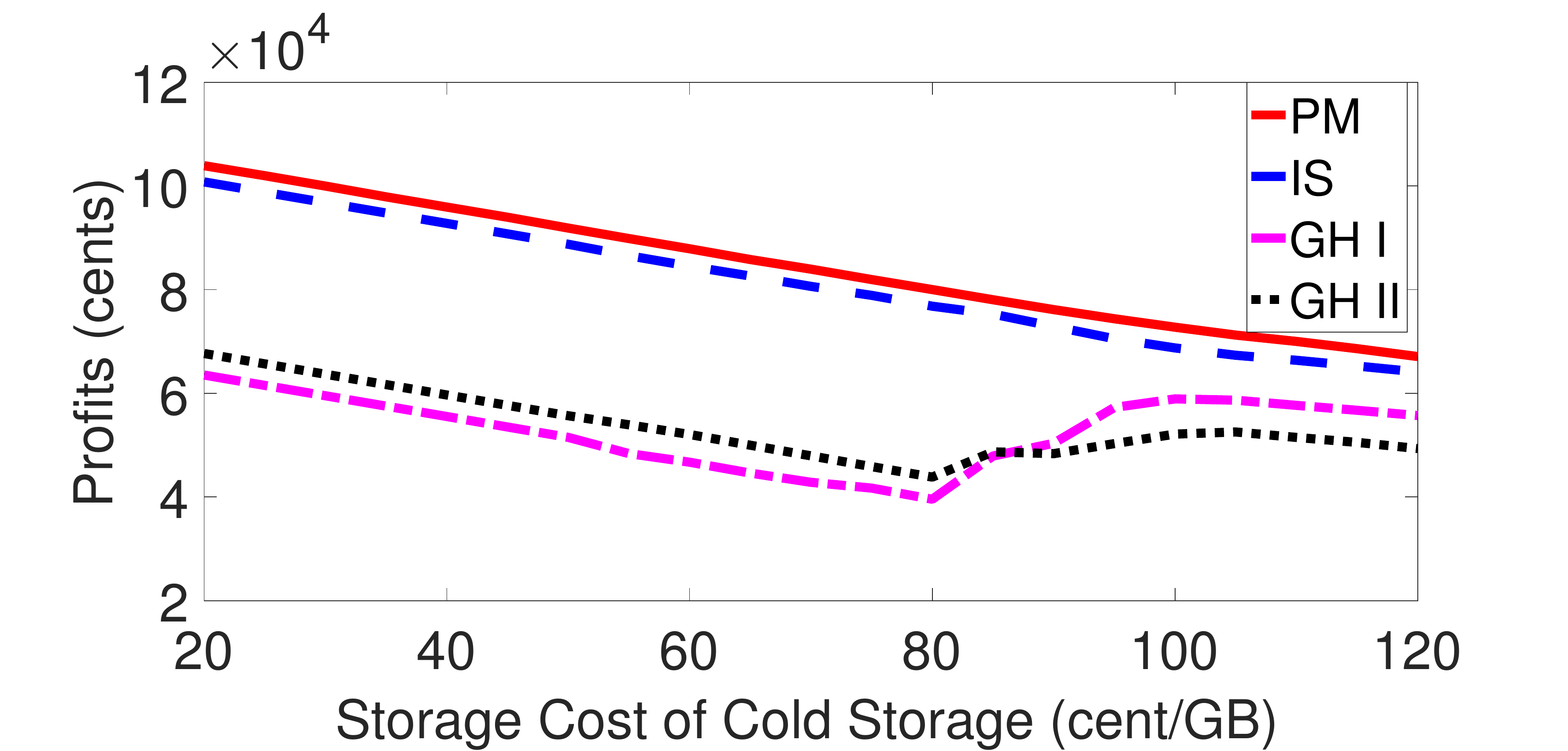}
   \label{fig:4.1}
 }
\subfigure[Impact of  cold storage cost on profits obtained from the file storage and access]{
   \includegraphics[scale =0.21] {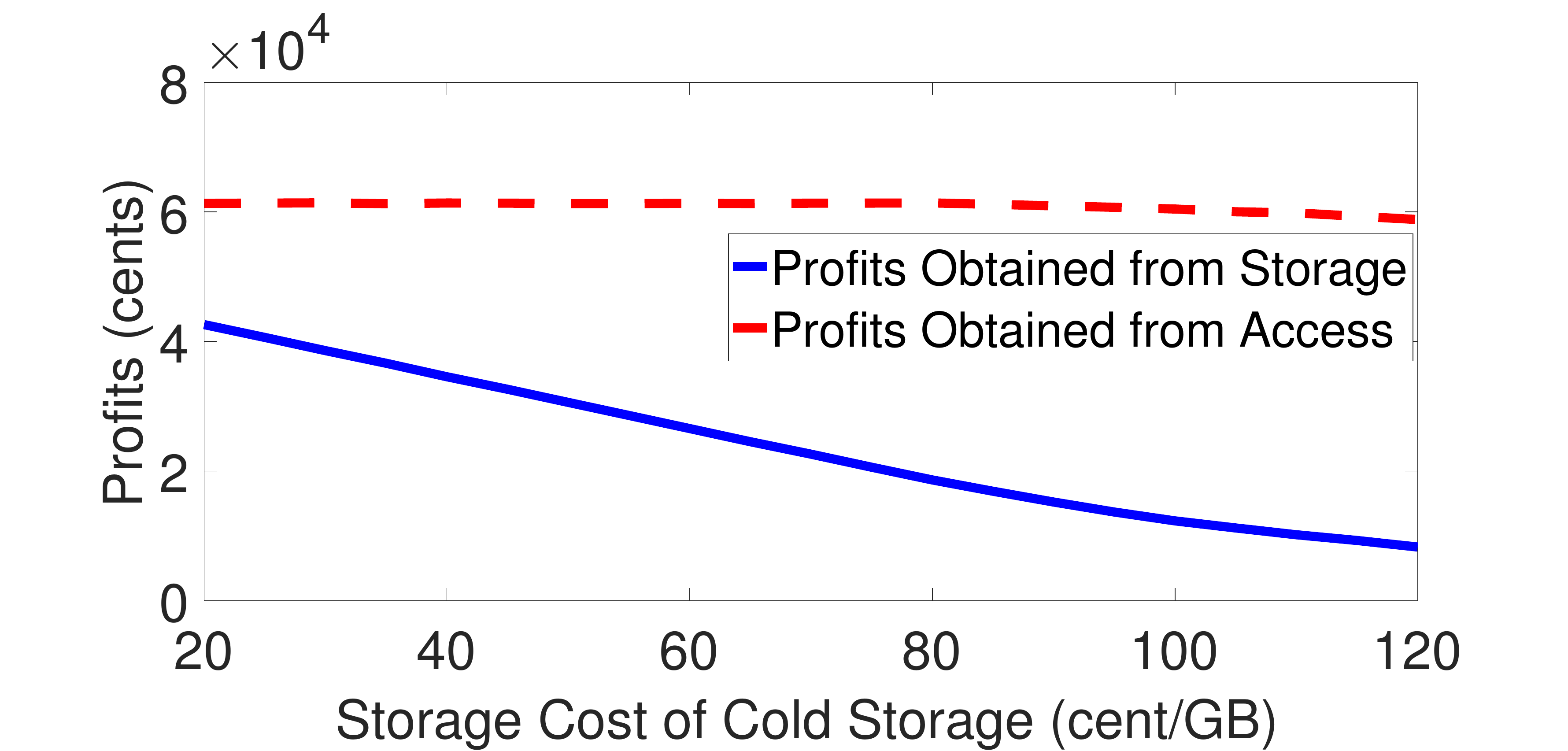}
   \label{fig:4.2}
 }
 \subfigure[Impact of cold  storage cost on ARAR]{
   \includegraphics[scale =0.21] {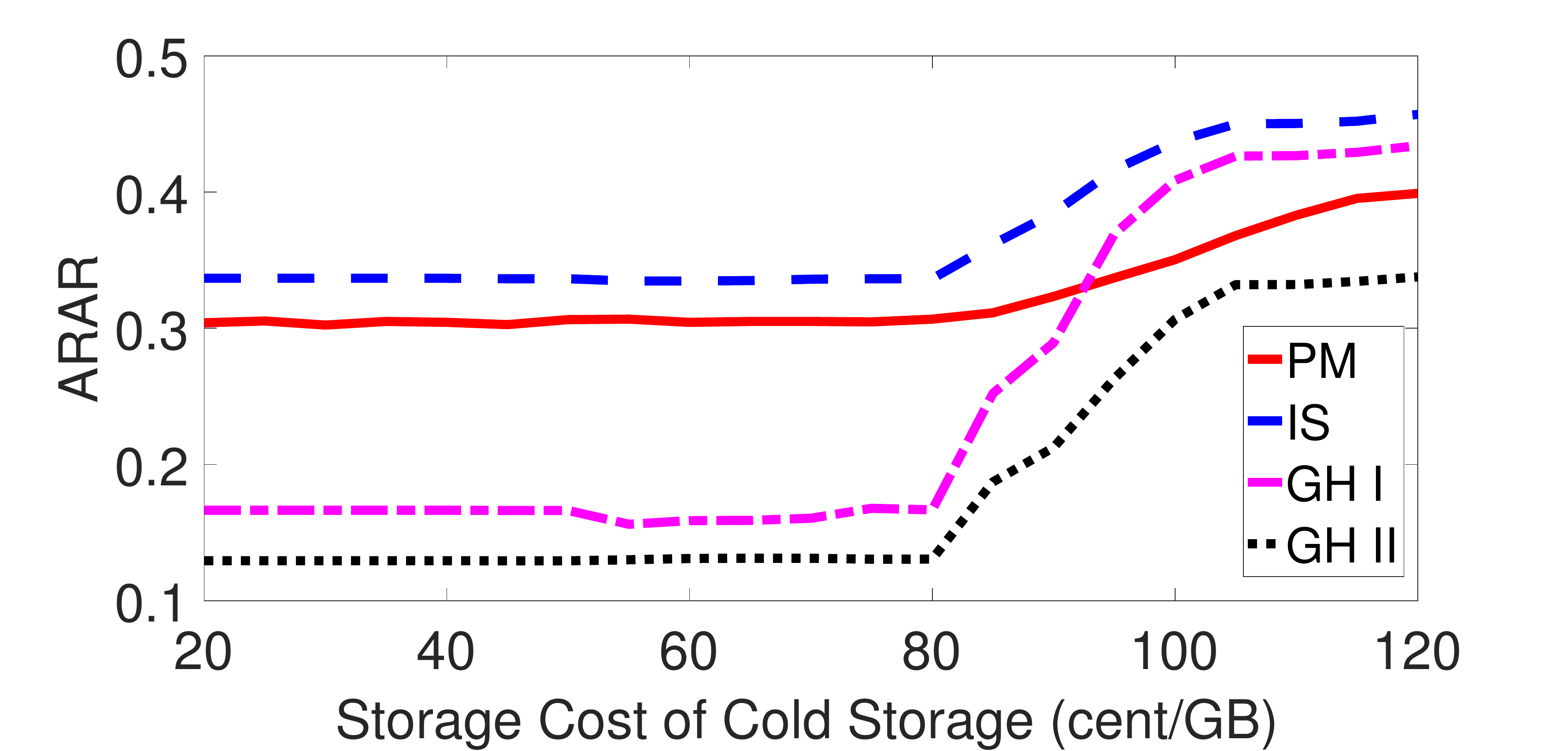}
   \label{fig:4.3}
 }
 \subfigure[Impact of cold storage cost on  the number of accepted bids for storage and access]{
   \includegraphics[scale =0.21] {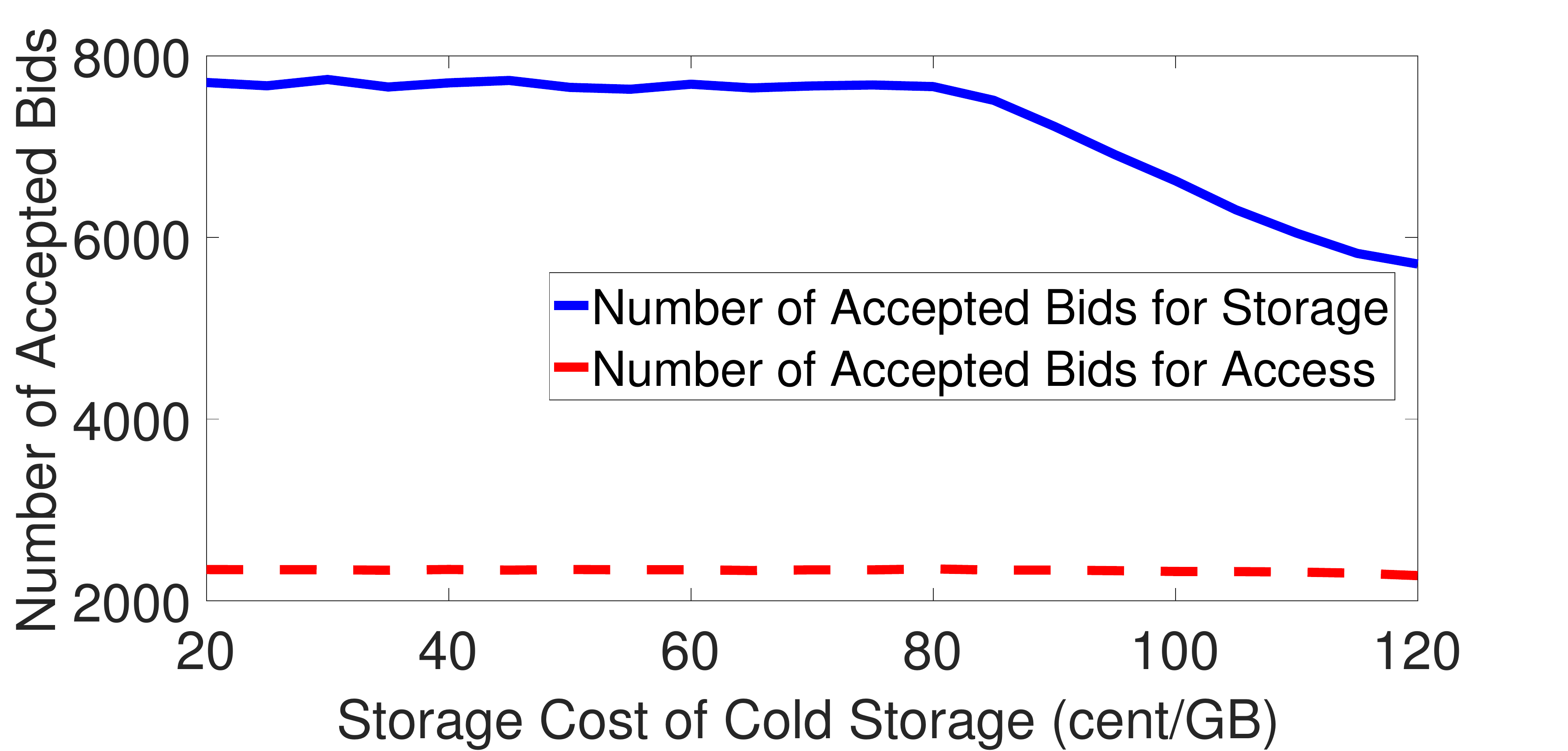}
   \label{fig:4.4}
 }
\caption{\small Profit and the ARAR as a function of the cold storage cost. $C_1=400$GB and $C_2=200$GB. The service rates of the cold and hot storages are $100$Gb/s and $200$Gb/s.}
\vspace{-.1in}
\label{fig4}
\end{figure}

In this subsection, we assume that the storage cost of hot storage is $80$ cents per GB, and vary the storage cost of cold storage  from $20$ cents to $120$  cents per GB in  steps of $5$ cents.  Fig.~\ref{fig:4.1} shows that as the cold storage cost increases the profit attained by our proposed method decreases.  As Fig.~\ref{fig:4.4} shows that when the cost of the cold storage increases the lower number of files are stored. Hence, the profit from the storage decreases (Fig.~\ref{fig:4.2}). The profit from accepting the access bids remain the same as the number of files stored in the hot storage almost remains the same.  Note from Fig.~\ref{fig:4.1} the profits gained from GH I and GH II decrease first and increase dramatically when the cold storage cost is $80$ cents per Gb (i.e., the hot storage cost). The main reason behind this is that beyond this point the hot storage has lower storage cost but higher service rate after that point, which means that the files are prioritized to be stored in the hot storage rather than the cold storage. Thus, profits from accepting the access bids increase drastically for those greedy heuristics and become close to the optimal.  

Note that as the cold storage cost exceeds the hot storage cost fewer number of files are stored. However, the accepted access bids remain the same as depicted in Fig.~\ref{fig:4.4}. Hence, the denominator of (\ref{ass}) decreases without decreasing the numerator. Thus, ARAR increases. Note that the IS does better in terms of ARAR. The greedy heuristic GH I also gives a higher ARAR when the cold storage cost exceeds the hot storage cost. This is because in the first-stage decision the larger files are mostly now stored in the hot storage  rather than the cold storage as they pay more. However, the total number of files accepted for cold and hot storage decreases. The larger files are also likely to bid higher in the second stage, thus, the ratio in (\ref{ass}) increases for GH I as it accepts bids in the descending order of the bids per size. However, GH II accepts bids in a different manner; hence, the ARAR attained by the GH II is strictly lower. 

\section{Conclusions and Future Work}\label{sec:concl}

In order to store the files in the {\em cold storage} or {\em hot storage}, this paper propose a systematic framework for two-stage, latency-dependent bidding, which aims to maximize the cloud storage provider's net profit in tiered cloud storage systems where tenants may have different budgets, access patterns and performance requirements. In the proposed two-stage, latency-aware bidding mechanism, the users can bid for storage and access, in two separate stages, without knowing how the CSP stores the contents. The proposed optimization is modeled as a mixed-integer nonlinear program (MINLP), for which  an efficient heuristic is proposed. The numerical results demonstrate that the profits obtained from the proposed method are higher than those of other methods, and the access request acceptance rate (ARAR) also dominates that of other methods as the capacity of the cold storage or the service rate of hot storage increases.

\textcolor{black}{
In reality, the users may be strategic. In other words, the user may optimize the bid in order to maximize its own profit. Our model captures some essence of the strategic users. The user's maximum possible bid will be the price that it will pay if it selects another CSP for the same guaranteed latency requirement. Thus, our approach can be easily extended to the above scenario where we can consider a upper limit of the bid of a user. We, however, did not consider the full essence. For example, the users may not bid truthfully. We will consider such a scenario in future.}

Another interesting direction for the future is to extend the model for erasure coding storage system where multiple copies ($n$) of the files can be stored and a subset of those copies ($k$) are required to be fetched to get the original file. In that case, the CSP would need to select the $n$ storage systems to store the file and among those $k$ copies are needed to be fetched to get the original file.

\section{Acknowledgment}
The authors would like to thank Yu Xiang and Robin Chen of AT\&T Labs-Research for helpful discussions.  This work was supported in part by the National Science
Foundation under Grant no. CNS-1618335.


%

\appendices
\section{Moments of the service time of a file request}\label{apdx:moments}
In this Appendix, we will derive the first and second moments of the service time of a file request at storage $j$ in $k$-th  scenario, which will be used further to prove Theorem  \ref{lm:waitingtime}. More precisely, we will show the following result. 

\begin{lemma}\label{lm:statistics}
$X_{j}^{k}$, the service time of a file request at storage $j$ in $k$-th  scenario,has a distribution with mean
\begin{equation} \label{eq3}
\mathrm{E}[ X_{j}^{k}]= \frac{ \sum_{i}\lambda_{i}^{k}\pi_{i,j}^{k}S_{i} }{ \mu_{j} \sum_{i}\lambda_{i}^{k}\pi_{i,j}^{k}  }
\end{equation}
and second moment
\begin{equation} \label{eq4}
\mathrm{E}[(X_{j}^{k})^{2}]= \frac{2 \sum_{i}\lambda_{i}^{k}\pi_{i,j}^{k}S_{i}^{2} }{\mu_{j}^{2} \sum_{i} \lambda_{i}^{k} \pi_{i,j}^{k}  }.
\end{equation}
\end{lemma}

The rest of the Section proves this result. 

It is easy to verify that under our model, the arrival of file requests at storage $j$ in $k$-th scenario forms a Poisson Process with rate $ \Lambda_{j}^{k} =\sum_{i} \lambda_{i}^{k}\pi_{i,j}^{k}$, which is the superposition of $I$ Poisson Processes each with rate $\lambda_{i}^{k}\pi_{i,j}^{k}$.

Let $S^{r}$ be the (random) requested file size at storage $j$, which is a discrete random variable such that the probability of $S^{r}=S_{i}$ is $\frac{\lambda_{i}^{k}\pi_{i,j}^{k}}{\sum_{i}\lambda_{i}^{k}\pi_{i,j}^{k}}$. Let $V^{t}_{j}$ be the (random) service time of one MB at storage $j$, which is exponentially distributed with mean $\frac{1}{\mu_{j}}$. The the expectation of the service time of a file request at storage $j$ is

\begin{equation} \label{eq1}
\begin{split}
 \mathrm{E}[X_{j}^{k}] & = \mathrm{E}_{V^{t}_{j}}[ \mathrm{E}_{S^{r}}[X_{j}|S^{r}=S_{i} ]  ]  \\
 & =\sum_{S^{r}}\mathrm{E}_{V^{t}_{j}}[X_{j}|S^{r}=S_{i}]\mathrm{P}\{S^{r}=S_{i}\}\\
 & =\sum_{S^{r}}\mathrm{E}_{V^{t}_{j}}[S^{r}V^{t}_{j}|S^{r}=S_{i}]\mathrm{P}\{S^{r}=S_{i}\}\\
 & = \sum_{i}\frac{S_{i}}{\mu_{j}}\frac{\lambda_{i}^{k}\pi_{i,j}^{k}}{\sum_{i}\lambda_{i}^{k}\pi_{i,j}^{k}}\\
 & = \frac{ \sum_{i}\lambda_{i}^{k}\pi_{i,j}^{k}S_{i} }{ \mu_{j} \sum_{i}\lambda_{i}^{k}\pi_{i,j}^{k}  }
\end{split}
\end{equation}

and the associated second moment is
\begin{equation} \label{eq2}
\begin{split}
 \mathrm{E}[(X_{j}^{k})^{2}] & = \mathrm{E}_{V_{j}^{t}}[ \mathrm{E}_{S^{r}}[X_{j}^{2}|S^{r}=S_{i} ]  ]  \\
 & =\sum_{S^{r}}\mathrm{E}_{V_{j}^{t}}[X_{j}^{2}|S^{r}=S_{i}]\mathrm{P}\{S^{r}=S_{i}\}\\
 & =\sum_{S^{r}}\mathrm{E}_{V^{t}_{j}}[(S^{r})^{2}(V^{t}_{j})^{2}|S^{r}=S_{i}]\mathrm{P}\{S^{r}=S_{i}\}\\
 & = \sum_{i}S_{i}^{2}\frac{2}{\mu_{j}^{2}}\frac{\lambda_{i}^{k}\pi_{i,j}^{k}}{\sum_{i}\lambda_{i}^{k}\pi_{i,j}^{k}}\\
 & = \frac{2 \sum_{i}\lambda_{i}^{k}\pi_{i,j}^{k}S_{i}^{2} }{\mu_{j}^{2} \sum_{i} \lambda_{i}^{k} \pi_{i,j}^{k}  }.
\end{split}
\end{equation}



%
%

\ifCLASSOPTIONcaptionsoff
  \newpage
\fi



%

\bibliographystyle{IEEEtran}
\bibliography{acmsmall-sample-bibfile}

\end{document}